\documentclass[11pt]{article}

\PassOptionsToPackage{table,xcdraw}{xcolor}
\usepackage{amsmath,amsthm,amssymb}
\usepackage{graphicx,psfrag,epsf}
\usepackage{enumerate}
\usepackage{multirow,booktabs}
\usepackage{url} % not crucial - just used below for the URL 
\usepackage{natbib}
\usepackage{makecell}
\usepackage{booktabs}
\usepackage{tabularx}
\usepackage{multibib}
\usepackage{multirow}
\newcites{refsupp}{Supplementary references}
\usepackage{longtable}
\usepackage{tikz}
\usetikzlibrary{positioning}
\usetikzlibrary{calc}
\usepackage{bm}
\usepackage{float}
\usepackage{hyperref}
\usepackage{geometry}
\usepackage{csquotes}
\usepackage{subcaption}
\usepackage{booktabs}
\usepackage{colortbl}
\usepackage{siunitx}
\usepackage[section]{placeins}

\usepackage{xr}
\newcommand*{\addFileDependency}[1]{% argument=file name and extension
\typeout{(#1)}% latexmk will find this if $recorder=0
\IfFileExists{#1}{}{\typeout{No file #1.}}
}\makeatother

\makeatletter
\AtBeginDocument{%
	\expandafter\renewcommand\expandafter\subsection\expandafter{%
		\expandafter\@fb@secFB\subsection
	}%
}
\makeatother

\usepackage{etoolbox}
\makeatletter
\patchcmd{\@makecaption}
{\parbox}
{\advance\@tempdima-\fontdimen2} % decrease the width!
{}{}
\makeatother 

\usepackage{geometry}
\usepackage{titling}
\makeatletter
{\begin{titlepage}\def\@thanks{}}%
{\end{titlepage}}
\makeatother

\newtheorem{theorem}{Theorem} 
\newtheorem{lemma}{Lemma}  
\newtheorem{corollary}{Corollary}
\newtheorem{assumption}{Assumption}

\theoremstyle{definition}

\theoremstyle{remark}

\providecommand{\keywords}[1]{\textbf{Keywords:} #1}

\title{ {\huge {Selection of functional predictors and smooth coefficient estimation for scalar-on-function regression models}  }\\ \vspace{-1mm}}

\makeatletter
\def\@author{
  \parbox{\textwidth}{\centering
    {\large {Hedayat Fathi}} \\ 
    \small{ Department of Operations and Decision Systems, Universit\'e Laval, Québec, Canada} \\ 
    \textcolor{blue}{hedayat.fathi.1@ulaval.ca} \\[2mm]
    
    {\large {Marzia A. Cremona}} \\ 
    \small{ Chair in Statistical Learning and Department of Operations and Decision Systems, Universit\'e Laval, Québec, Canada} \\ \textcolor{blue}{marzia.cremona@fsa.ulaval.ca} \\[2mm]
    
    {\large {Federico Severino}} \\ 
    \small{ Department of Finance, Insurance and Real Estate, Universit\'e Laval, Québec, Canada} \\ 
    \textcolor{blue}{federico.severino@fsa.ulaval.ca} \\[2mm]
    }
}
\makeatother

\date{}

\begin{document}

\maketitle

%%%%%%%%%%%%%%%%%%%%%%%%%%%%%%%%%%%%%%% ABSTRACT

\begin{abstract}
In the framework of scalar-on-function regression models, in which several functional variables are employed to predict a scalar response, we propose a methodology for selecting relevant functional predictors while simultaneously providing accurate smooth (or more generally regular) estimates of the functional coefficients. 
We suppose that the functional predictors belong to a real separable Hilbert space, while the functional coefficients belong to a specific subspace of this Hilbert space. Such a subspace can be a Reproducing Kernel Hilbert Space (RKHS) to ensure the desired regularity characteristics, such as smoothness or periodicity, for the coefficient estimates. 
Our procedure, called SOFIA (Scalar-On-Function Integrated Adaptive Lasso), is based on an adaptive penalized least squares algorithm that leverages functional subgradients to efficiently solve the minimization problem. We demonstrate that the proposed method satisfies an appropriate version of the oracle property adapted to the functional setting, even when the number of predictors exceeds the sample size. 
SOFIA's effectiveness in variable selection and coefficient estimation is evaluated through extensive simulation studies and a real-data application to predict GDP growth. 
\end{abstract}



%\newpage

%\tableofcontents

%\newpage

\section{Introduction} 
\label{sec:intro}
High-dimensional regression models are crucial in modern data analysis and have numerous applications in several fields such as genetics, finance, and the social sciences. In the extreme case, some or all of the variables included in the model are functions that vary over a continuum, such as time or space. The use of functional data \citep{kokoszka2017} can result in more accurate data analysis, improved pattern extraction, and enhanced visualization. However, functional data analysis (FDA) introduces new challenges in statistical modeling. For example, in functional regression models with a large number of predictors, selecting the most relevant ones is needed to improve model interpretability and predictive accuracy, and to reduce computational complexity \citep[see, e.g,~][]{chen2020}.   

In this study, we focus on \emph{scalar-on-function regression} involving a scalar response and several functional predictors. %and introduce a new approach to simultaneously perform variable selection and smooth coefficient estimation. 
In particular, we consider the standardized regression model
\begin{equation}\label{S-O-F}
Y_i = \sum_{j=1}^{p} \left\langle \beta^{\ast}_j, X_{ij} \right\rangle_{\mathbb{H}} + \varepsilon_i, \quad i = 1,\cdots, n,
\end{equation}
where $Y_i$ is the scalar response, $X_{ij}$ are the $p$ functional predictors belonging to a real separable Hilbert space $(\mathbb{H}, \langle \cdot,\cdot \rangle_{\mathbb{H}})$, %(in the simplest case $L^2[0, 1]$), 
$\varepsilon_i$  are independent scalar random errors with mean 0 and variance $\sigma^2$, and $n$ is the sample size. For each $j$, the function $\beta^{\ast}_j$ quantifies the effect of the $j^{\text{th}}$ functional predictor on the scalar response. Applications of scalar-on-function regression are widespread, including finance (where curves are time series), medicine, and the natural sciences \citep{ullah2013}. %See \cite{ullah2013} for a systematic review of applications of functional linear regression and, more broadly, FDA. 
Although variable selection methods for classic linear regression models are ubiquitous in the statistical literature, %\citep[see two recent examples in][]{insolia2022,hanke2024}, 
they have only recently been extended to the functional setting \citep[][]{aneiros2022}.

%In this paper, we 
We focus on situations where several functional predictors are present but only a small number among them are active (i.e.,~have $\beta^{\ast}_j \neq 0$). 
Even in classical multiple regression models, %(where no functional variable is involved), 
when the number of predictors is large, best subset selection may be computationally infeasible and stepwise selection provides only locally optimal solutions to the variable selection problem. Lasso can be effective in practice \citep{tibshirani1996}, but does not satisfy the oracle property \citep{fan2001}. 
The adaptive lasso method \citep{zou2006} adds a data-driven weighting scheme to the penalty term of lasso to improve the accuracy of the estimates and ensure the validity of an oracle property
%. Under appropriate assumptions, the adaptive lasso is an oracle estimator even 
when the number of predictors increases with the sample size \citep{huang2008}.
Following this literature,
we propose an adaptive lasso method for scalar-on-function regression named \emph{SOFIA (Scalar-On-Function Integrated Adaptive Lasso)}, that simultaneously performs variable selection and coefficient function estimation, ensuring a desired level of regularity for the coefficients. 
First, we employ a functional version of (non-adaptive) lasso to filter out some variables and obtain initial estimates of the coefficients for the selected predictors. Second, similarly to adaptive lasso, we set weights as reciprocals of the absolute values of the obtained estimates. %Alternatively, one can choose the weights in a subjective or ad hoc manner based on domain knowledge or prior information about the importance of the predictors. However, this approach may not always lead to optimal performance in practice. 
The primary challenge in the scalar-on-function setting is that the design matrix does not consist of scalar entries but contains elements of a Hilbert space. %This difference introduces several computational and theoretical challenges. 
%
%To avoid redundancy in terminology, we refer to our method as \emph{SOFIA-Lasso} (Scalar-On-Function Integrated Adaptive Lasso), or simply \emph{SOFIA}. The term \emph{integrated} reflects the integration of two norms, enabling the simultaneous selection of variables and regularization of functional coefficients. 
%From an asymptotical point of view, 
SOFIA enjoys an asymptotic oracle property even in the case where the number of variables and the number of active predictors increase with the sample size. Specifically, we allow the number of predictors to grow faster than the sample size (even at an exponential rate). The challenges of such situations in multiple regression models are illustrated in \cite{huang2008}. 

In our framework, while predictors belong to a general Hilbert space \(\mathbb{H}\), we assume that coefficient functions \(\beta^{\ast}_j\) belong to a specific dense Hilbert subspace \(\mathbb{K}\) that incorporates desirable regularity properties such as smoothness or periodicity. This subspace is defined by the eigenfunctions of a positive self-adjoint operator $K$ on \(\mathbb{H}\). In the special case when a positive-definite kernel induces the operator $K$, the resulting space forms a reproducing kernel Hilbert space \citep[RKHS; see][]{hsing2015}. 
This framework pursues two primary objectives. First, by choosing an appropriate operator $K$, we ensure the desired level of regularity for the coefficient functions $\beta^{\ast}_j$. 
%This provides tailored regularization strategies to address the specific needs of various applications. 
Second, it results in a spectral decomposition that facilitates the approximation of the objective function within finite-dimensional subspaces.
%Second, this assumption, together with the decay properties of the eigenvalues of the operator $K$, facilitates the approximation of the objective function within finite-dimensional subspaces while ensuring that the approximation error remains uniformly bounded. %Indeed, the infinite-dimensional nature of the problem presents a significant theoretical challenge in scalar-on-function regression: the covariance matrix lacks a bounded inverse, and the conditions required to guarantee asymptotic properties are rarely satisfied. 
%
\(\mathbb{K}\) is a dense subspace of \(\mathbb{H}\) equipped with a stronger norm, and both share the same orthogonal basis, given by the eigenfunctions of \(K\). 
%Hence they are defined on the same algebraic structure, but the topology of \(\mathbb{K}\) is finer and induces a stronger norm than that of \(\mathbb{H}\).}
This distinction allows us to define %the objective function of SOFIA and 
the following penalized least squares estimator of the coefficients in \eqref{S-O-F}:  
%\begin{equation}\label{model:1}
%\hat{\boldsymbol{\beta}}= \underset{\boldsymbol{\beta}\in \mathbb{K}^p}{\operatorname{argmin}} ~ \frac{1}{2n} ~\sum_{i=1}^{n}\left( Y_i - \sum_{j=1}^{p} \left\langle \beta_j, X_{ij}  \right\rangle_{\mathbb{H}} \right) ^2 + \lambda \sum_{j=1}^{p} \tilde{w}_j \left\| \beta_j \right\|_{\mathbb{K}},
%\end{equation}
\begin{equation}\label{model:1}
\hat{\boldsymbol{\beta}}
= \underset{\boldsymbol{\beta}\in \mathbb{K}^p}{\operatorname{argmin}}
\Bigg\{
\frac{1}{2n}\sum_{i=1}^{n}
\left(
Y_i - \sum_{j=1}^{p} \left\langle \beta_j, X_{ij} \right\rangle_{\mathbb{H}}
\right)^2 + \lambda \sum_{j=1}^{p} \tilde{w}_j \left\| \beta_j \right\|_{\mathbb{K}}
\Bigg\},
\end{equation}
where \(\lambda\) is a tuning parameter and \(\tilde{w}_j\) are adaptive weights designed to penalize smaller coefficient norms more heavily.

A crucial challenge in functional regression is the infinite-dimensional nature of the problem. This issue is particularly pronounced in scalar-on-function regression, where the empirical covariance matrix contains finite-rank operators on $\mathbb{H}$ and
%Consequently, the dimension of the underlying space always exceeds the rank of the covariance matrix, and so this matrix (and its restrictions to the space of true predictors) 
cannot be invertible as an operator. 
%To address this issue, a 
A common approach 
is the so-called \textit{sieve method} \citep{wang2016}, which consists in projecting the predictors onto a finite-dimensional subspace %of $\mathbb{H}$, denoted by 
$\mathbb{H}^{(m)}$. This subspace is defined by an orthonormal basis, which may be predefined (e.g., splines), data-driven (e.g., functional principal components), or derived from a general basis \citep{roche2023}. 
%This type of finite-dimensional approximation is commonly known as \textit{sieve method} \citep{wang2016}. 
In the sieve method, the functional entities in the model are approximated over a sequence of nested finite-dimensional subspaces whose union is dense in $\mathbb{H}$. %the underlying Hilbert space. 
As the sample size increases, the sieve subspace expands, enabling the estimator to approximate the true functional parameters more accurately. This highlights the nonparametric nature of the problem, where consistent estimation is achieved by letting the complexity of the model (the sieve dimension $m$) grow with the sample size rather than fixing a finite number of parameters.
Following this method, we derive the required basis from the eigenfunctions of the operator \( K \) and we project both functional predictors and coefficient functions onto the subspaces $\mathbb{H}^{(m)}$ and $\mathbb{K}^{(m)}$, respectively, and to obtain the following %finite-dimensional 
approximation of \eqref{model:1} on $\mathbb{K}^{(m)}$:  

\begin{equation}\label{model:2}
\hat{\boldsymbol{\beta}}^{(m)}
= \underset{\boldsymbol{\beta}\in \mathbb{K}^{(m)p}}{\operatorname{argmin}}
\Bigg\{
\frac{1}{2n}\sum_{i=1}^{n}
\left(
Y_i - \sum_{j=1}^{p}
\left\langle \beta_j, X_{ij}^{(m)} \right\rangle_{\mathbb{H}}
\right)^2 
 + \lambda \sum_{j=1}^{p} \tilde{w}_j \left\| \beta_j \right\|_{\mathbb{K}}
\Bigg\}.
\end{equation}

\subsection{Related works}

The problem of variable selection for scalar-on-function regression models with multiple functional predictors was first examined by \cite{cardot2007ozone}, which employed best subset selection in the context of ozone prediction. Although a few methods have been proposed since then, the literature on the subject remains relatively limited. Most of these existing methods are based on the use of a truncated basis expansion %a group modeling strategy \citep{aneiros2022} which allows replacing 
to replace functional variables with a group of scalars, %variables, and hence reformulating 
and reformulate the model as a multiple regression model with grouped predictors \citep{aneiros2022}.
%Briefly, each functional variable is expanded using a pre-defined basis system (e.g.,~B-splines or Fourier basis) and the expansion is then truncated to obtain a finite-dimensional representation of the variable. 
A grouped variable selection method \citep[see, e.g.,][]{yuan2006} is then employed to select groups of predictors, which correspond to functional variables.
The first two works proposing this strategy are \cite{matsui2011}, which use Gaussian basis functions combined with a penalized least squares approach, and \cite{gertheiss2013}, which employ B-spline basis functions with an adaptive group lasso penalty for variable selection. In the latter, the penalty includes two components to regulate the sparsity of the solution and enforce smoothness, respectively.  
\cite{smaga2018note} propose a basis-expansion approach in which functional predictors are selected using a random subspace method combined with information criteria.
While these methods are interesting from an algorithmic point of view, they do not provide theoretical guarantees on the selection of the true active variables. 

Among the methods %based on a group modeling strategy 
that provide theoretical guarantees on variable selection, %performance, 
we find \cite{fan2015}, which use a generic basis and a group lasso penalty, \cite{collazos2016}, which assume a finite B-spline representation of the predictors and the coefficients, and \cite{lian2013} and \cite{huang2016}, which expand the functional data through functional principal components (FPCA) and use penalized least squares and penalized least absolute deviation, respectively.
\cite{fan2015} provides a rigorous asymptotic analysis, also when the number of predictors exceeds the sample size. However, the assumption that predictors are known to belong to a Sobolev ellipsoid is quite restrictive. 
\cite{huang2016} %propose a similar method employing FPCA to expand functional predictors and a different penalty term. 
% ($\ell_1/\ell_{\infty}$ norm instead of $\ell_1/\ell_2$). 
%While this approach 
yields similar analytical results, but %it 
is more computationally complex, since it requires estimating a set of basis functions for each functional predictor.
Finally, \cite{lian2013} provides asymptotic results in the case of fixed number of predictors and of active predictors. 
More recently, \cite{boschi2024} introduced FAStEN, an adaptive elastic net method for variable selection %in function-on-function regression, which can be applied also to the scalar-on-function case. This work also 
which expands the functional data using FPCA and proves an asymptotic oracle property. Unlike FPCA-based approaches, SOFIA does not require the approximation basis to come from the covariance operator nor does it impose explicit conditions on covariance eigenfunctions. The basis system is obtained from the operator $K$, which is independent of the data.

All the variable selection methods mentioned above only consider the case of functional data belonging to $L^2(T)$, where $T$ is a bounded interval in $\mathbb{R}$. To the best of our knowledge, \cite{roche2023} is the only work considering predictors in general Hilbert spaces and proving asymptotic and non-asymptotic properties.
In particular, it proposes two adaptations of the group lasso that satisfy sharp oracle inequalities for prediction errors. However, it does not prove results on the selection of the true active variables nor on the coefficient estimations. %, which we prove for SOFIA. 
 
RKHS have received considerable attention over the past two decades in functional linear regression. A seminal contribution is represented by \cite{yuan2010}, who introduced a smoothness regularization method based on RKHS. Since then, RKHS-based approaches have been employed in various FDA studies for different purposes \citep{kadri2010,wang2022}. 
%For the foundational aspects of RKHS in functional regression models, we refer to \cite{kadri2010}; recent advancements can be found in \cite{wang2022}. 
In the context of RKHS-based linear regression, two approaches can be used for parameter estimation, namely the \emph{representer theorem} and \emph{functional subgradients}. The former is the dominant approach in the literature due to its theoretical convenience \citep{yuan2010,berrendero2019}. However, the methods based on the representer theorem tend to be computationally slow \citep{parodi2018}. For this reason, SOFIA is based on functional subgradients \citep[Chapter 16]{bauschke2011}. 
In a very recent contribution, \cite{GuoLiHsing2028} proposed Functional Elastic-Net (fEnet), a method based on RKHS for scalar-on-function regression models with functional data belonging to $L^2$. Although fEnet and SOFIA share the objectives of performing variable selection and obtaining regularized coefficients, the two methods differ fundamentally. The model in \cite{GuoLiHsing2028} is an elastic net for a random design with Gaussian noise, and prediction consistency is obtained by post-selection regularization. On the contrary, we introduce an adaptive lasso-type method for a fixed design with sub-Gaussian noise and a penalty term derived from a topology induced by a compact operator. We do not require a second step for coefficient regularization, and we perform variable selection and coefficient regularization simultaneously in an adaptive way. Our proposed method, both theoretically and in practice, is applicable to a general Hilbert space that can be chosen to suit the user's needs, and we do not require a bivariate kernel to control the level of smoothness or regularity; any appropriate compact kernel works for the method. More importantly, we employ the sieve method both in theory and in implementation. The method of \cite{GuoLiHsing2028} is designed in $L^2$ and uses the representer theorem for theoretical results, while the implementation uses truncation based on the eigenfunctions of the empirical covariance operators of the kernel-transformed predictors, due to the computational cost of the representer theory (which avoids basis expansion). Since fEnet uses the covariance operators for basis expansion in its implementation, the smoothness of the coefficients is not directly controlled, while we use the operator $K$ to give full control over the smoothness or regularity of the coefficients.

Finally, SOFIA takes inspiration from a series of papers by Reimherr and coauthors on lasso-based variable selection %based on lasso regularization. While these papers explore 
in the case of function-on-scalar regression, %as opposed to scalar-on-function regression, 
from which we borrow %from them 
the definition of the functional oracle property and of the Hilbert space \(\mathbb{K}\). \cite{barber2017} propose FS-Lasso, a method based on the (unweighted) group lasso. Although the oracle property is not satisfied, they give some reasonable bounds for estimation and prediction errors. \cite{fan2017} propose an adaptive version of the penalty term in \cite{barber2017} and prove the %(functional) 
oracle property of the resulting method. Finally, FLAME \citep{parodi2018} satisfies the oracle property for variable selection and guarantees regularity %, while ensuring desired properties such as smoothness or periodicity 
for the coefficients by requiring that they belong to an RKHS.
%Finally, in \cite{mirshani2021}, they introduced an elastic net method to leverage the advantages of both the Lasso and Ridge methods.

\subsection{Our contributions}
The general objective of this paper is to introduce a method for variable selection in scalar-on-function regression models that ensures the regularity of the estimated coefficients while satisfying the functional oracle property. The theory is designed for general separable Hilbert spaces, and the regularity is controlled by a compact operator with spectral decomposition. A distinctive feature of
SOFIA is that the eigensystem of this operator is used not only to
control coefficient regularity, but also to construct the sieve spaces
underlying both the theoretical analysis and the implementation. Unlike FAStEN \citep{boschi2024} and fEnet \citep{GuoLiHsing2028}, which rely on an FPCA-based representation, SOFIA constructs its finite-dimensional approximation from the eigensystem of the user-specified operator $K$, thereby providing direct control over the smoothness, or more generally, the regularity of the estimated coefficient functions.

Our first theoretical result concerns the estimator of \eqref{model:2} in the (non-adaptive) case of \(\tilde{w}_j=1\) and states that, under a suitable restricted eigenvalue (RE) condition, the estimation error admits a sharp bound. %This bound provides a rigorous starting point for the adaptive step. 
Although \cite{barber2017} established a similar result, we provide a novel proof for predictors that are finite-dimensional truncations of functional data. %, i.e$.$ vectors.
Our main result concerns %theorem (Theorem~\ref{thm:main}) proves 
the functional oracle property of SOFIA in the spirit of \cite{fan2017}, i.e., proves
%That is, 
that SOFIA correctly identifies the true predictors with probability tending to one as %the sample size 
$n$ increases. %, under some appropriate assumptions on the finite-dimensional projections.  
If the tails of the true coefficients after truncation decay rapidly, then the estimation error declines with the nonparametric rate $o_P(\tau_m/n^{1/2})$, where $\tau_m$ is the minimum eigenvalue of the empirical covariance 
%restricted to the finite-dimension subspace 
in $\mathbb{K}^{(m)}$. %Importantly, 
%Since this result is in the topology of $\mathbb{K}$ and this topology is stronger than that of $\mathbb{H}$, it
This result also holds in the topology of $\mathbb{H}$. This solves %responds to 
an open problem %question 
of \cite{parodi2018}.
 
Another major contribution is the extensive comparison of the variable selection and estimation performance of SOFIA with multiple recent methods, %in different simulation scenarios. In particular, we compare SOFIA with 
namely the standard group lasso \citep{fan2015}, the semi-adaptive and adaptive group lasso \citep{gertheiss2013}, the method of \cite{roche2023}, FAStEN \citep{boschi2024}, and fEnet \citep{GuoLiHsing2028}.  
Finally, we provide a real-world example for predicting Canadian GDP growth. %using several economic and financial variables. %This application presents several challenges, such as handling unbalanced data and high-frequency variables, which we address by representing the predictors as functional objects. 

%\subsection{Organization of paper}

%Section~\ref{sec:back} introduces the framework and assumptions about the data and coefficients, and provides the mathematical notation and technical conditions needed for the theoretical analysis. Section~\ref{sec:resul} presents the main theoretical results. Section~\ref{sec:impl} details the implementation procedure based on functional subgradients. Section~\ref{sec:sim} presents an extensive simulation study and compares SOFIA's performance with existing approaches.  Section~\ref{sec:real} provides a real-data application to the prediction of GDP growth. Finally, Section \ref{sec:concl} concludes the work.

\section{Mathematical framework}\label{sec:back}

%In this section, we describe the general structure of the Hilbert space in which the functional predictors are defined and present the model assumptions, including regularity conditions on the coefficient functions and the structure of the trace-class operator that defines the Hilbert subspace where the coefficients reside. We then provide a structured set of assumptions that ensure the uniform convergence of the estimation procedure.

%\subsection{Underlying spaces}
We assume that the predictors $X_{ij}$ are elements of real separable Hilbert space $(\mathbb{H},\|\cdot\|_{\mathbb{H}}) $, where the norm is derived from the inner product $\langle \cdot,\cdot \rangle_{\mathbb{H}}$. Note that $\mathbb{H}$ encompasses a wide range of spaces, including $L^2[0,1]$ (used in the simulations of Section~\ref{sec:sim}), $\mathbb{R}^N$, and more specialized function spaces commonly used in statistics, such as product spaces, Sobolev spaces, Wiener spaces, and RKHS.

To ensure regularization of the coefficient functions $\beta_j^*$, we follow the approach of \cite{parodi2018}. We consider a (compact) trace-class operator, $K :\mathbb{H} \longrightarrow \mathbb{H}$, which is positive definite and self-adjoint. %i.e., $\langle Kx, x\rangle_{\mathbb{H}} \geq 0$ and $\langle Kx, y\rangle_{\mathbb{H}} = \langle x, Ky\rangle_{\mathbb{H}}$ for all $x$ and $y$ in $\mathbb{H}$. 
%Every trace class operator is compact, see for example \cite{conway2000} Theorem 18.11 
The spectral theorem for compact self-adjoint operators implies that $K$ can be decomposed as $K = \sum_{l=1}^{\infty} \theta_l v_l \otimes v_l,$ where $ \theta_1 \geq \theta_2 \geq \dots > 0$ are the ordered eigenvalues and $v_l \in \mathbb{H}$ are the corresponding eigenfunctions, which constitute an orthonormal basis for $\mathbb{H}$ \citep[Theorem 2.4.4]{Murphy1990}. The notation $x \otimes y$ denotes the rank-one operator defined by $(x \otimes y)(h) := \langle y, h \rangle_{\mathbb{H}} x$ for any $x,y,h \in \mathbb{H}$.  
%A large class of operators satisfies these assumptions, including compact operators and many integral operators. 
We define %the subspace $\mathbb{K}$ of $\mathbb{H}$ as
\begin{equation}\label{spec}
\mathbb{K}   := \left\lbrace h \in \mathbb{H} :~ \sum_{l=1}^{+\infty} \frac{\left|\left\langle h, v_l \right\rangle_{\mathbb{H}}\right|^2}{\theta_l} =: \left \|h \right \|_{\mathbb{K}}^2  < {+\infty} \right\rbrace,
\end{equation}
which is a Hilbert space with norm $\|\cdot\|_{\mathbb{K}}$. %, $\mathbb{K}$. 
When \( \mathbb{H} = L^2[0,1] \), many examples can be found in the literature. In particular, if \( K \) is defined as an integral operator associated with a kernel function \( k(t,s) \) in a compact domain, % by definition, all kernels are symmetric and positive definite. They are Hilbert-Schmidt operators, so they are always compact, but being trace class requires a compact domain.  
then \( K \) is a trace-class, positive, and self-adjoint operator. As a result, any RKHS can be viewed as a special case of \( \mathbb{K} \). 
%In Section~\ref{sec:sim}, we will introduce a family of operators to implement our method. 
However, in general, $K$ need not be induced by a bivariate integral kernel. It can be specified directly through an orthonormal basis $\{v_l\}_{l=1}^{\infty}$ of $\mathbb{H}$ and a positive non-increasing summable sequences $\{\theta_l\}_{l=1}^{\infty}$. We illustrate this operator-based construction in a Sobolev space in Section~S2.6. %of the Supplementary material.
For a detailed discussion on compact and trace-class operators in Hilbert spaces, we refer to \cite{Murphy1990}, Section~2.4. For an overview of various types of kernels, see \cite{berlinet2011}, Chapter~7.
We assume that the coefficient functions $\beta_j^*$ belong to $\mathbb{K}$, ensuring the desired level of regularity of the coefficients by choosing an appropriate operator $K$. 
The objective function of SOFIA, shown in \eqref{model:1}, utilizes the $\mathbb{K}$-norm in the penalty term, while the least squares term involves the inner product in $\mathbb{H}$. We study the relationship between $\mathbb{K}$ and $\mathbb{H}$ in Section~S1.1.  
Intuitively, we show that $\mathbb{K}$ is the range of the operator $K^{1/2}$, with an appropriate topology, and its norm can be characterized by $\|x\|_{\mathbb{H}} = \|K^{1/2}(x)\|_{\mathbb{K}}$. We also show that $\mathbb{K}$ is dense in $\mathbb{H}$ under the norm of $\mathbb{H}$.

\subsection{Model}
The following assumption establishes the foundational framework for our model by specifying the structure of the response, the predictors, and the coefficient functions.
\begin{assumption} \label{Assumption 1}
Let $Y_i$ denote the scalar response for subject $i$,  and $X_i = (X_{i1}, \ldots, X_{ip}) \in \mathbb{H}^p$ represent the deterministic functional predictors. 
Without loss of generality, we assume that $Y_i$ is centered and the predictors standardized, i.e.,~$\sum_{i=1}^{n} Y_i = 0$, $\sum_{i=1}^{n} X_{ij} = 0$, and $n^{-1} \sum_{i=1}^{n} \|X_{ij}\|^2_{\mathbb{H}} = 1$, $j = 1,\cdots, p$.
We assume that they satisfy the functional linear model \eqref{S-O-F} with coefficient functions $\boldsymbol{\beta}^* = (\beta^*_1, \ldots, \beta^*_p) \in \mathbb{K}^p \subset \mathbb{H}^p$. The errors $\varepsilon_i$ are i.i.d.\ random variables with mean zero and variance $\sigma^2$. In addition, for constants $C_1, C_2 > 0$, we assume that $P\{|\varepsilon_i| > t\} \leq C_1 \exp(- C_2 t^2)$ for all $t \geq 0$ and $i=1,2,\ldots,n$. 
%This assumption is taken from \cite{huang2008}.
\newline Moreover, we assume that only the first $p_0 < p$ predictors are active, while the remaining $p-p_0$ predictors do not contribute to $Y_i$ (i.e.,~%. Consequently, 
their corresponding coefficients are zero). 
\end{assumption}

We note that the sub-Gaussian tail condition on the errors is satisfied, for instance, by Gaussian random variables \citep[Section 2.3]{boucheron2003}.
Importantly, in our model the number of predictors $p$ and the number of active predictors $p_0$ are not fixed, but can increase with the sample size $n$. In particular, we require $p_0$ to be much smaller than both $n$ and $p$, while we allow $p$ to grow even exponentially with $n$ in our asymptotic results. 
%Our framework is very general, covering $\mathbb{H}=L^2[0,1]$ (used in Section~\ref{sec:sim}), $\mathbb{H} = \mathbb{R}$ for classical multiple regression, and additional instances where $\mathbb{H}$ is an RKHS or a Sobolev space. 

%\citep{fan2017} 
%If $\mathbb{H} = \mathbb{H}^1_2[0, 1]$, i.e., the Sobolev space of real-valued functions over the unit interval with one square-integrable derivative, then 
%\begin{equation*}
%Y_i = \sum_{j=1}^{p}\left(\int_{0}^{1} \beta_j^{\ast}(t)X_{ij}(t) dt + \int_{0}^{1} \beta_j^{\ast\prime}(t)X_{ij}^{\prime}(t) dt \right) + \varepsilon_i.
%\end{equation*}

\subsection{Finite-dimensional approximation}
In infinite-dimensional Hilbert spaces, the covariance operator is typically non-invertible due to its finite rank. To address this issue, we restrict the covariance operators to finite-dimensional subspaces of \(\mathbb{K}\). In particular, the \emph{sieve method} \citep{wang2016} enables us to approximate our infinite-dimensional estimator of \eqref{model:1} by solving a sequence of finite-dimensional problems. Following \cite{fan2015}, we consider a sequence of integers \( m = m_n = O(n)\), which depends only on the sample size \(n\). As $n$ grows to infinity, the dimension $m$ also tends to infinity and the sieve subspaces become dense in $\mathbb{K}$, allowing the estimator to converge to the true coefficient functions. Let \(\mathbb{K}^{(m)}\) denote the \(m\)-dimensional subspace of \(\mathbb{K}\) spanned by \(\{v_1, \cdots, v_{m}\}\), and \(\mathbb{K}^{(m)\perp}\) its orthogonal complement. The projections of \(\beta_{j} \in \mathbb{K}\) on \(\mathbb{K}^{(m)}\) and \(\mathbb{K}^{(m)\perp}\) are denoted by \(\beta_{j}^{(m)}\) and \(\beta_{j}^{(m)\perp}\), respectively. We use the analogous notation for the projections of \(X_{ij}\) onto \(\mathbb{H}^{(m)}\) and \(\mathbb{H}^{(m)\perp}\). %The properties of \(m\) and the nested finite-dimensional Hilbert subspaces \(\{\mathbb{K}^{(m)}\}_{m}\) will be discussed in detail later. 
In order to approximate the functional coefficients and predictors while controlling the approximation errors, we impose standard assumptions on their bounds.

\begin{assumption} \label{Assumption 2}
There exist constants $b_1, b_2 > 0$ and $\gamma > 1$ such that \( p_0 = o(m^{(\gamma -1)/2}) \), and the active coefficients and predictors satisfy:
\[ 
\left\|\beta_j^{\ast} \right\|_{\mathbb{K}} \leq b_1 \quad \text{and} \quad \sup_{1 \leq i \leq n} \theta_l \left\langle X_{ij}, v_l \right\rangle_{\mathbb{H}}^2 \leq b_2 l^{-\gamma}
\]
for all  $j = 1, \ldots, p_0$ and $l \geq 1$.
\end{assumption}

The assumption on the true coefficients implies that such coefficients are uniformly bounded,  a standard assumption in functional regression models \citep{fan2015, GuoLiHsing2028}. Since we work with deterministic predictors, the condition on the active predictors is satisfied, for example, when 
$\|X^{[j]}\|_{\mathbb{H}}$, the column of the design matrix related to the $j^\text{th}$ predictor, is uniformly bounded %i.e., when there exist $C>0 such that $\sup_{i,j}\|X_{ij}\|_{\mathbb{H}}\le C_$
and the operator $K$ exhibits at least a polynomial decay rate ($\theta_l \leq C l^{-\gamma}$ for some constants $C > 0$ and $\gamma > 1$). The operator $K$ meets this assumption, for example, when it is induced by the Gaussian kernel or the Matérn kernels. Indeed, the Gaussian kernel has exponential decay, while the Matérn kernel has a decay rate of the form  $\theta_l \leq C l^{-(1+2\nu)}$ \citep{widom1964}, hence the assumption is satisfied for any smoothing parameter $\nu > 0$.
Since our assumption directly concerns the projection coefficients of active predictors on a given basis, we do not need additional assumptions on the eigenstructure of the covariance operators, which are required in FPCA-based methods \citep{lian2013, fan2015, boschi2024}.
A condition similar to the bounded design is used in \cite{roche2023}, where a deterministic design matrix is considered.  The condition \( p_0 = o(m^{(\gamma -1)/2}) \) reflects the fact that we are working in a sparse regime, where the number of active functional predictors grows more slowly than $m^{(\gamma-1)/2}$, where $m$ is the sieve dimension and $\gamma$ the spectral decay. For kernels with fast spectral decay (e.g., Gaussian kernel or Matérn kernel with $\nu>1$), this requirement is really mild, whereas for kernels with a slower decay rate, such as exponential/Laplacian kernels, it corresponds to the sparsity condition \( p_0 = o(\sqrt{m})\). % because exponential is the same as Matern kernel with $\nu = 1/2.$
 Finally, in finite samples, this asymptotic condition does not require \(m > p_0\) for a fixed sample size. The behavior of SOFIA when \(m < p_0\) and \(m > p_0\) is examined in Section~S2.4 (Table~S2). % of Supplementary material

The predictors and coefficients can be expanded on the basis \(\{v_l\}_{l=1}^{+\infty}\) as
%\[
$X_{ij} = \sum_{l=1}^{+\infty} \langle X_{ij}, v_l \rangle_{\mathbb{H}} v_l$, %\qquad 
and 
$\beta_j^{\ast} = \sum_{l=1}^{+\infty} \langle \beta_j^{\ast}, v_l \rangle_{\mathbb{H}} v_l$.
%\]
If we restrict them %the predictors and coefficients 
to the $m$-dimensional subspaces, we obtain the following decomposition for our regression model \eqref{S-O-F}:
\begin{equation}\label{approx}
Y_i = \sum_{j=1}^{p_0} \left\langle \beta_j^{*(m)} , 
  X_{ij}^{(m)} \right\rangle_{\mathbb{H}} + \sum_{j=1}^{p_0} e_{ij} + \varepsilon_i,
\end{equation}
where \( e_{ij} \) denotes the approximation error, defined as
\(
e_{ij} = \sum_{l=m+1}^{\infty} \langle X_{ij}, v_l \rangle_{\mathbb{H}} \langle \beta_j^{\ast}, v_l \rangle_{\mathbb{H}}.
\)
Proposition~S2 of Section~S1.2 %of the Supplementary material 
shows that, under Assumption~\ref{Assumption 2}, these approximation errors converge uniformly to zero as $m \to +\infty$. This permits us to consider the finite-dimensional truncation of model~\eqref{S-O-F}:
\begin{equation}\label{fini}
Y_i = \sum_{j=1}^{p_0} \left\langle \beta_j^{*(m)} , X_{ij}^{(m)}  \right\rangle_{\mathbb{H}} + \varepsilon_i.
\end{equation}

To estimate the finite-dimensional coefficients, we consider the objective function of SOFIA: $L_{\lambda}: \mathbb{K}^{(m)} \longrightarrow \mathbb{R}^+,$
\begin{equation}\label{lass-fini}
L_{\lambda}(\boldsymbol{\beta}) = \frac{1}{2n} \sum_{i=1}^{n} \left(Y_i - \sum_{j=1}^p \left\langle  \beta_j , X_{ij}^{(m)} \right\rangle_{\mathbb{H}} \right)^2 + \lambda \sum_{j=1}^p \tilde{w}_j \| \beta_j \|_{\mathbb{K}}.
\end{equation}
The weights \( \tilde{w}_j \) are chosen adaptively. First, we set all weights to 1 to perform an initial selection and obtain an estimate \( \tilde{\beta}_j \). Second, we perform an adaptive step with weights \( \tilde{w}_j = 1/ \|\tilde{\beta}_j\|_{\mathbb{K}} \). %This approach ensures that significant coefficients remain nearly unweighted while smaller ones are forced to shrink toward zero.

%Next, we introduce some notation. We 
In the following, we use bold symbols to denote vectors and matrices. Let $\boldsymbol{X} = (X_{ij}) \in \mathbb{H}^{n \times p}$ be the Hilbert-valued deterministic design matrix, which we partition as $\boldsymbol{X} = (\boldsymbol{X}_1, \boldsymbol{X}_2)$, where $\boldsymbol{X}_1 \in \mathbb{H}^{n \times p_0}$ and $\boldsymbol{X}_2 \in \mathbb{H}^{n \times (p - p_0)}$ represent the submatrices of active and null predictors, respectively. The support of $\boldsymbol{\beta}^{\ast}$ is defined as $\mathcal{A} = \{ 1, \ldots, p_0 \}$ and its complement is $\mathcal{A}^{c}$. When there is no ambiguity, we denote the inner product and norm of both $\mathbb{K}^p$ and $\mathbb{K}$ using the same notation.

Since the least squares estimator is based on the inner product of $\mathbb{K}$, we define the covariance operator between two predictors as an operator on $\mathbb{K}$. For \( j, k = 1, \dots, p_0 \), the empirical covariance operators are
\begin{equation}\label{cov}
    \hat{\Gamma}_{j k} : \mathbb{K} \longrightarrow \mathbb{K}, \quad \hat{\Gamma}_{j k} \left( h \right) = \frac{1}{n} \sum_{i=1}^n \left\langle h, K\left(X_{ik}\right) \right\rangle_{\mathbb{K}} K\left(X_{ij}\right).
\end{equation}
These operators incorporate the projection of predictors in the range of \( K \). In addition, we define the Hilbert-valued cross-covariance for \( j = 1, \dots, p_0 \) as
\(
\hat{\Gamma}_{Y X_j} = n^{-1} \sum_{i=1}^n Y_i K\left(X_{ij}\right) \in \mathbb{K}.
\)
We denote by \( \hat{\Gamma}_{\boldsymbol{X}_1} \) the \( p_0 \times p_0 \) matrix of operators \( \hat{\Gamma}_{j k} \) and define the vector of cross-covariance operators as \( \hat{\Gamma}_{Y\boldsymbol{X_1}} = (\hat{\Gamma}_{Y X_1}, \dots, \hat{\Gamma}_{Y X_{p_0}})^\top \). The cross-covariance operator-valued matrix between null and active predictors is denoted by \( \hat{\Gamma}_{\boldsymbol{X_2X_1}} \) and is defined analogously.
For $j,k = 1, \cdots, p$, the finite-dimensional covariance operators, indicated by \( ^{(m)} \), are $\hat{\Gamma}_{j k}^{(m)} : \mathbb{K}^{(m)} \longrightarrow \mathbb{K}^{(m)},$
\[
\hat{\Gamma}_{j k}^{(m)} (h) = \frac{1}{n} \sum_{i=1}^n \left\langle h, K\left(X_{ik}^{(m)}\right) \right\rangle_{\mathbb{K}} K\left(X_{ij}^{(m)}\right).
\]

%\subsection{Oracle estimator and oracle property}
Finally, we focus on the oracle estimator. This estimator, denoted by \( \hat{\boldsymbol{\beta}}^O = (\hat{\beta}_1^{O}, \dots, \hat{\beta}_p^{O}) \), is typically defined as the least squares estimator for active predictors, while it is set to zero for inactive ones. However, in scalar-on-function regression, the empirical covariance matrix \( \hat{\Gamma}_{\boldsymbol{X_1}} \) is always singular, making the least squares estimator inaccessible. It should be noted that we regard the $p_0 \times p_0$ matrix of operators $\hat{\Gamma}_{\boldsymbol{X}_1}$ as a linear operator on $\mathbb K^{p_0}$. 
%acting as $\hat{\Gamma}_{\boldsymbol{X}_1}\boldsymbol{h}_j = \sum_{k=1}^{p_0} \hat{\Gamma}_{jk}(h_k)$. 
Its invertibility is understood in this operator sense \citep[see, e.g.,][Section~3.4]{Murphy1990}.
In our theory, we work on finite-dimensional spaces $\mathbb{K}^{(m)}$ and define the sieve-oracle estimator as
$\boldsymbol{\hat{\beta}}_1^{O(m)} = (\boldsymbol{\hat{\beta}}_1^{O(m)}, 0)$ where
\(
\boldsymbol{\hat{\beta}_1}^{O(m)} = \hat{\Gamma}_{\boldsymbol{X_1}}^{(m)-1} \hat{\Gamma}_{Y\boldsymbol{X_1}}^{(m)}.
\)
In Section \ref{sec:resul}, we demonstrate that our estimator satisfies the functional oracle property %. Then, we show that under appropriate assumptions, it 
and tends to the true beta at a rate determined by the minimum eigenvalue of the empirical covariance matrix. 

\subsection{Restricted eigenvalue (RE) condition}

In our approach, as a first step, we set all the weights in the objective function~\eqref{lass-fini} to 1. The obtained estimator does not achieve the oracle property, similar to the scalar case \citep{fan2001}. Nevertheless, in Theorem~\ref{thm:lasso}, we show that under an appropriate version of the RE condition, a tight bound can be found for the distance between the true and the estimated coefficients. We use this bound in the adaptive step to prove the oracle property. 

When the predictors are functional, the classic definition of the RE condition is never met: when the dimension of the underlying space is infinity, the lower bound in the RE condition is always zero \citep[Lemma 2.1]{roche2023}. To overcome this issue, \cite{roche2023} proposes to define the RE condition through a sequence \({\alpha}^{(m)}_{n}\) of nonnegative (possibly zero) values. We adopt a similar approach and consider the following condition.

\begin{assumption}(\textbf{RE condition})\label{Assumption RE}
For the dimension \(m\), let
\[
\mathcal{C}_m = \left\{
\begin{array}{l}
\boldsymbol{\delta}=(\delta_1,\dots,\delta_p)\in \mathbb{K}^{(m)p}\setminus\{0\} : \\[2pt]
\exists\, J\subseteq \{1,\dots,p\},\ |J|\le p_0,\\
\|\boldsymbol{\delta}_{J^c}\|_{\ell_1/\mathbb{K}}
\le 3\|\boldsymbol{\delta}_J\|_{\ell_1/\mathbb{K}}
\end{array}
\right\}.
\]
There exists a positive constant \(\alpha_m\) such that
\[
\left\|\Gamma_{\boldsymbol{XX}}^{(m)1/2}\boldsymbol{\delta}\right\|_{\mathbb{K}}
\ge
\alpha_m \|\boldsymbol{\delta}\|_{\mathbb{K}},
\qquad \forall\, \boldsymbol{\delta}\in \mathcal{C}_m,
\]
where $\|\boldsymbol{\delta} \|_{ \mathbb{K}} = ( \sum_{j=1}^p \| \delta_j \|_{\mathbb{K}}^2)^{1/2}$ represents the $\ell_2$ norm, $\| \boldsymbol{\delta}_J\|_{\ell_1 / \mathbb{K}} =  \sum_{j\in J} \| \delta_j \|_{\mathbb{K}}$ and $J^c$ denotes the complement of the set $J$.
\end{assumption}

Note that our RE condition differs from that of \cite{roche2023} in two key aspects. First, %in our definition, 
we construct the covariance operators using the approximated predictors instead of the predictors. Second, we assume that the dimension $m$ depends on the sample size $n$ and is smaller than it, while \cite{roche2023} treats the dimension \(m\) and the sample size \(n\) as independent, allowing \(m\) to grow to infinity while \(n\) remains fixed. As a result, in \cite{roche2023} for sufficiently large dimensions \(\alpha^{(m)}_{n}\) can become zero.
%which is problematic for our theoretical analysis. 
In contrast, in our setting \(m\) depends on \(n\), and \(m = O(n)\). In this way, the lower bounds $\alpha^{(m)}$ in Assumption \ref{Assumption RE} can form a sequence of strictly positive scalars, typically tending to zero.
%For more details on the RE condition, see \cite{bickel2009} for lasso, \cite{lounici2011} for group lasso, \cite{bellec2017} for a general convex penalty, and \cite{barber2017} for function-on-scalar models.

\subsection{Asymptotic assumptions}

We use the following technical assumptions, which are well-documented in the literature, to assess the asymptotic behavior of our model. If the design matrix is scalar-valued, the upper and lower bounds are typically assumed to be constant \citep[see][]{huang2008,fan2017,parodi2018,mirshani2021}. In contrast, following \cite{roche2023}, we assume that these bounds are sequences that tend to zero as $m$ increases. We study the relation between the RE condition and Assumption~\ref{Assumption assy} in Proposition~S3 in Section~S1.3. 
Specifically, we show that when the covariance between the null predictors tends to zero rapidly, the RE condition is met under Assumption~\ref{Assumption assy}.

\begin{assumption} \label{Assumption assy}
For the dimension $m$,  the following conditions hold.
\begin{enumerate}
%\setcounter{enumi}{-1}
 %\item Assume that the dimension $m = m_N$ grows slower than the sample size $N$. Specifically, we assume that $m = o(N)$. DO WE NEED TO SAY HERE? 
\item Let $\sigma_{\min}$ and $\sigma_{\max}$ be the smallest and largest eigenvalues of $\hat{\Gamma}_{\boldsymbol{X_1}}^{(m)},$  respectively. There exist a fixed constant $\tau$ and a sequence $\tau_{m} > 0$ such that
\[
\frac{1}{\tau_{m}}\leq \sigma_{\min}\left(\hat{\Gamma}_{\boldsymbol{X_1}}^{(m)}\right) \leq \sigma_{\max}\left(\hat{\Gamma}_{\boldsymbol{X_1}}^{(m)}\right) \leq \tau.
\] 
 \item  Let $b = \min_{j \in \mathcal{A}} \|\beta_j^{\ast}\|_{\mathbb{K}}$,~ $b_{m} = \min_{j\in \mathcal{A}^{(m)}}  \|\beta_j^{\ast(m)}\|_{\mathbb{K}}$, and $\zeta_m = \min \{\tau_m^{-1} , \alpha_m^2\}$. Then,  
 \[
 \underset{j\in \mathcal{A}}{\max}\|\beta_j^{*} - \beta_j^{*(m)}\|_{\mathbb{K}} = o(b), \quad \text{and} \quad b_{m} \gg  p_0 \sqrt{\log(p)}/\zeta_m \sqrt{n}.
 \]
 \item The tuning parameter $\lambda$ satisfies the bounds \\
\( p_0^{1/2} \log(p)\zeta_m^{-1} n^{-1} \ll \lambda \ll b_{m} p_0^{-1/2} n^{-1/2}. \)
\item Let $\|\cdot\|_{\text{op}}$ denote the operator norm. Then, \\
 \(\|\hat{\Gamma}_{\boldsymbol{X_2 X_1}}^{(m)} \hat{\Gamma}_{\boldsymbol{X_1}}^{^{(m)}-1}\|_{\text{op}} = \varphi_{m} < 1.\)
%\item \(\zeta_m \gg p_0^2 \log(p)n^{-1}.\)
\end{enumerate}

\end{assumption}

These assumptions are counterparts to some standard assumptions in the high-dimensional regression literature. The first one, often called \emph{design matrix condition} in the scalar design, ensures that the true predictors exhibit regular behavior and are not excessively correlated. This condition is imposed by setting fixed lower and upper bounds for the eigenvalues. In our approach, since \( \hat{\Gamma}_{\boldsymbol{X_1}} \) is a finite-rank operator on an infinite-dimensional Hilbert space \( \mathbb{H}^{p_0} \), it is never injective, and so \( 0 \) is always an eigenvalue. To address this issue, we restrict our analysis to finite-dimensional subspaces of the underlying infinite-dimensional space, making it natural to impose a sequence of bounds that depend on both the dimension and the sample size. Typically, \( (\tau_{m})^{-1} \) has zero as an accumulation point, since, as the sample size increases, \( \hat{\Gamma}_{\boldsymbol{X_1}}^{(m)} \) converges to a compact operator (not necessarily of finite rank), and compact operators usually have \( 0 \) as a limit point of their eigenvalues (unless they are of finite rank). This implies that the design matrix becomes increasingly ill-conditioned and nearly singular as $n$ grows. This ill-posedness is controlled by the assumptions on the minimum norm of the coefficients and the tuning parameter. In contrast, when the population covariance $\Gamma_{\boldsymbol{X_1}}$ is bounded, the upper eigenvalue bound can be taken as fixed, similar to scalar cases.
The second group of conditions corresponds to the \emph{minimum signal condition}. It plays two roles in our study. First, it ensures that sieve truncation does not erase the active coefficients. Second, it allows the smallest \(\beta_j^* \) to vary with the sample size, dimension, lower bound of the eigenvalues, and RE bound, while ensuring that it does not become too small. If \( (\tau_{m})^{-1} \) or \(\alpha_m^2\) decay too rapidly in relation to sample size, number of predictors, and number of active predictors, then there is no guaranty that the model is able to select predictors with small norms. 
%Furthermore, the decay rate of this sequence plays a crucial role in the variable selection performance of the model, as captured by the second condition, known as the \enquote{minimum signal condition}. 
The third condition prevents the tuning parameter from growing too rapidly or too slowly as $n$ grows. The lower bound prevents overfitting by controlling false positives, while the upper bound ensures that relevant predictors are not excessively penalized, avoiding false negatives. 
Finally, the fourth condition,  known as the \emph{irrepresentable condition}, ensures that the true and the null predictors do not exhibit a high correlation. In the literature on scalar predictors, the sequence $\varphi_m$ is supposed to be a fixed number, usually lower than 1. This requirement has been established as a necessary condition for having the oracle property \citep{zhao2006}.

\section{Theoretical results}\label{sec:resul}
This section presents the theoretical results of the paper. We begin with a concentration inequality that bounds the tail probability of a weighted sum of Hilbert space-valued random variables and serves as a key component in the proofs of Theorems~\ref{thm:lasso} and~\ref{thm:main}. This result is an extension of Lemma 1 in the Supplementary material of \cite{huang2008} to the Hilbert space setting. Its proof is given in Section~S1.4.% of the Supplementary material.

\begin{lemma}[\textbf{Hilbert-valued concentration inequality}] \label{lem:hilbert_concentration}
 Let \( \boldsymbol{X} = (X_{ij}) \in \mathbb{H}^{n\times p} \) and \( \boldsymbol{\varepsilon} = (\varepsilon_1, \cdots, \varepsilon_n) \) be as in Assumption~\ref{Assumption 1}. For each \( j = 1, \ldots, p \), define 
\(
Z_j := n^{-1} \sum_{i=1}^n \varepsilon_i X_{ij}
\). Then, for any \( t > 0 \), 
\[
P \left( \left\| Z_j \right\|_{\mathbb{H}} > t \right) \leq 2 \exp \left( - \frac{n t^2}{M} \right),
\]
where \( M > 0 \) is a constant depending only on the constants $C_1$ and $C_2$ in Assumption~\ref{Assumption 1}.
\end{lemma}

Theorem \ref{thm:lasso} provides the bounds for the estimation and prediction errors of the unweighted scalar-on-function lasso. A similar theorem for classic multiple regression can be found in Chapter 6 of \cite{buhlmann2011}. The function-on-scalar case is proved in \cite{barber2017}. Our contribution is two-fold. First, we design a concentration inequality for Hilbert-valued design matrices. Second, we develop the estimation and error analysis in a separable Hilbert space.
Our concentration inequality mirrors the scalar structure, hence the %non-asymptotic and asymptotic 
convergence rates are nonparametric counterparts of those of the scalar case. These rates are crucially dependent on the sequence $\{\alpha_{m}\}_{m\in \mathbb{N}}$ in Assumption~\ref{Assumption RE}. This sequence is supposed to be fixed in classic regression models, so it does not have any impact when the sample size increases. However, in the presence of functional predictors, this sequence always tends to zero and its behavior affects the convergence rate of the estimator.

\begin{theorem} \label{thm:lasso}
Let Assumptions~\ref{Assumption 1}-\ref{Assumption RE} hold. Let \( \tilde{\boldsymbol{\beta}}^{(m)} \) be any minimizer of the lasso function \eqref{lass-fini} with weights equal to 1. 
Let \( \delta \in (0,1) \) and  
\(
\lambda \geq 2 \sqrt{n^{-1} \theta_1 M \log(2p/\delta) },
\)
where $M >0 $ is the constant in Lemma~\ref{lem:hilbert_concentration}. Then,
\[
P\left( \left\| \Gamma_{\boldsymbol{X X}}^{(m)1/2} \left( \tilde{\boldsymbol{\beta}}^{(m)} - \boldsymbol{\beta}^{\ast (m)} \right) \right\|_{\mathbb{K}} \leq \frac{4 \lambda \sqrt{p_0}}{ \alpha_m } \right) \geq 1 - \delta  \]
and
\[
\quad P\left( \left\| \tilde{\boldsymbol{\beta}}^{(m)} - \boldsymbol{\beta}^{\ast (m)} \right\|_{\ell_1/\mathbb{K}} \leq \frac{8 \lambda p_0}{\alpha_m^2 }\right)\geq 1 - \delta.
\]
\end{theorem}

The proof of Theorem \ref{thm:lasso} is provided in Section~S1.5. % of the Supplementary material. 
Its usefulness is apparent in Corollary~\ref{cor.lasso}, which %is a direct result of Theorem~\ref{thm:lasso} and 
makes (unweighted) lasso a very good starting point for the adaptive model. In this corollary, the upper bounds are selected in a way that is consistent with the literature 
\citep[see, e.g.,][]{lounici2011,fan2017}.

\begin{corollary} \label{cor.lasso}
Let Assumptions~\ref{Assumption 1}-\ref{Assumption RE} hold, and $\lambda$ decays with rate \( \lambda \sim  \sqrt{\log(p)/p_0 n}. \)
Then, for any minimizer \( \tilde{\boldsymbol{\beta}}^{(m)} \) of the lasso function \eqref{lass-fini} with weights equal to 1, for $n\to+\infty$,
\[
\sup_{j \in \mathcal{A}^{(m)}} \left\| \tilde{\beta}_j^{(m)} - \beta_j^{\ast(m)} \right\|_{\mathbb{K}} 
= O_P \left( \sqrt{ r_m } \right)
~~ \text{with} ~~ r_m := \frac{  p_0 \log p }{ \alpha_m^4 n }.
\]
\end{corollary}

%\begin{proof}
%Note that the decay rate of $\lambda$ fulfills the requirement of Theorem~\ref{thm:lasso}. 
%Because $\lambda$ is always larger than the one in the theorem.
%Therefore, we have
%\begin{equation}\label{ine:tild}
 %\sup_{j \in \mathcal{A}^{(m)}} \left\| \tilde{\beta}_j^{(m)} - \beta_j^{\ast(m)} \right\|_{\mathbb{K}} \leq \frac{8 \lambda p_0} {\alpha_m^2}. \tag{3}   
%\end{equation}
%Now,
%\[
%\sup_{j \in \mathcal{A}^{(m)}} \left\| \tilde{\beta}_j^{(m)} - \beta_j^{\ast(m)} \right\|_{\mathbb{K}} \leq \frac{ p_0 \sqrt{\log(I)}} {\alpha_m^2 \sqrt{N}}   = \frac{ p_0 \sqrt{\log(I)}} {\alpha_m^2\sqrt{N}} .
%\]

%\end{proof}

Theorem \ref{thm:main} is our main result and shows that SOFIA asymptotically selects the true predictors and converges to the oracle estimator with probability tending to 1. 

\begin{theorem}[\textbf{Functional oracle property}] \label{thm:main}
 Let Assumptions~\ref{Assumption 1}-\ref{Assumption assy} hold. For $j = 1, \cdots, p$ let $\tilde{\beta}_j^{(m)}$ be the estimator of Theorem~\ref{thm:lasso} and $\hat{\beta_j}^{(m)}$  a minimizer of the function~\eqref{lass-fini} with weights $\tilde{w}_j  = \|\tilde{\beta_j}^{(m)}\|^{-1}_{\mathbb{K}}$. Then, when $n\to +\infty$,

(i) \(
P \left(\hat{\mathcal{A}}^{(m)} = \mathcal{A} \right) \to 1.
%P\left(\hat{\boldsymbol{\beta}}^{(m)} \stackrel{\mathcal{A}}{=} \boldsymbol{\beta}^{\ast(m)}\right) \rightarrow 1.
\)

(ii) \(
 \left\| \hat{\boldsymbol{\beta}}^{(m)} - \hat{\boldsymbol{\beta}}^{O(m)} \right\|_{\mathbb{K}} = o_P(\tau_m/\sqrt{n}).
\)

(iii) If, in addition, $ \underset{j \in \mathcal{A}}{\max}\|\beta_j^{\ast^{(m)\perp}}\|_{\mathbb{K}} = {o(\tau_m/\sqrt{p_0 n})}$, then, as $n\to +\infty$,  
\[
 \left\| \hat{\boldsymbol{\beta}}^{(m)} - \boldsymbol{\beta}^{\ast} \right\|_{\mathbb{K}} = o_P(
 \tau_m/\sqrt{n}
 ).
\]
\end{theorem}

The first part of the theorem establishes that SOFIA recovers the true support on the $(m)$-dimensional subspace, with probability tending to one. This ensures consistent variable selection in the high-dimensional functional setting. The second part shows that SOFIA is asymptotically equivalent to the oracle estimator in the finite-dimensional space $\mathbb{K}^{(m)}$. %This means that SOFIA is asymptotically equivalent to the oracle estimator on $\mathbb{K}^{(m)}$. 
The third part claims that, if the tails of the projections of $\beta_j^*$ in $\mathbb{K}^{(m)}$ decay sufficiently fast, then the finite-dimensional lasso estimator converges to $\boldsymbol{\beta}^{\ast}$ in the strong topology of $\mathbb{K}$. As a result, SOFIA achieves full consistency in $\mathbb{K}$, %under appropriate conditions, 
despite operating through finite-dimensional approximations.
The proof of Theorem~\ref{thm:main} given in S1.7, is inspired by \cite{huang2008} and \cite{parodi2018}, and relies on the functional subgradient of the objective function, together with a concentration inequality for sums of Hilbert-valued random elements. The latter yields high-probability bounds for the stochastic error term. Such bounds are then combined with the adaptive penalization framework to establish both support recovery and estimation consistency.

The following corollary is a direct result of Theorem~\ref{thm:main} when $p_0$ is fixed and $\tau_m$ grows at a polynomial rate, and gives a concrete convergence rate %by selecting the truncation level $m$ as a 
in function of $n$. %This leads to the usual bias-variance trade-off: the estimation term $\tau_m/\sqrt{n}$ increases with $m$, while the approximation error $\max_{j\in \mathcal{A}}\|\beta^{*(m)\perp}_j\|_{\mathbb{K}}$ decreases.

\begin{corollary} \label{cor:cor2-rate}
Assume the conditions of Theorem~\ref{thm:main} and suppose in addition that $p_0$ is fixed. Let $\tau_m \asymp m^{\eta}$ for some $\eta>0$, and assume
\(  \underset{j \in \mathcal{A}}{\max}\|\beta^{*(m)\perp}_j\|_{\mathbb{K}} = o(m^{-s}),\) for some \(s>0\).
If we choose the sieve dimension 
\( m \asymp n^{1/(2(s+\eta))}.\)
Then, as $n\to+\infty$,
\[
\left\| \hat{\boldsymbol{\beta}}^{(m)} - \boldsymbol{\beta}^{\ast} \right\|_{\mathbb{K}} = O_P \left( n^{-\frac{s}{2(s+\eta)}} \right).
\]
\end{corollary}

\section{Implementation}\label{sec:impl}

To implement SOFIA, we take inspiration from the FLAME method of \cite{parodi2018} for the function-on-scalar regression case.
%and introduce several modifications to account for functional predictors. 
To find the minimizers of the objective function~\eqref{lass-fini}, we propose a coordinate descent algorithm based on the functional subgradient \citep[Chapter 16]{bauschke2011}. In Section~S1.6 % of the Supplementary material
we show that, for $j = 1,\ldots, p$, the functional subgradient of \eqref{lass-fini} in $\mathbb{K}^{(m)}$ is 
\begin{equation*}
\begin{aligned}
\frac{\partial}{\partial \beta_j}\! L_\lambda(\boldsymbol{\beta})
& = \!-\frac{1}{n}
\sum_{i=1}^n
\biggl( Y_i \!-\! \sum_{k=1}^p
\left\langle \beta_k, K\left(X_{ik}^{(m)}\right) \right\rangle_{\mathbb{K}}
\biggr)
K\left(X_{ij}^{(m)}\right) \\
&\quad + \lambda \tilde{w}_j
\begin{cases}
\left\|\beta_j\right\|_{\mathbb{K}}^{-1}\beta_j,
& \text{if } \beta_j \neq 0, \\[4pt]
\left\{ h \in \mathbb{K}^{(m)} : \left\|h\right\|_{\mathbb{K}} \le 1 \right\},
& \text{if } \beta_j = 0.
\end{cases}
\end{aligned}
\end{equation*}
To update the $j^{th}$ coordinate, we set 
$\check{\beta}_j = n^{-1} \sum_{i=1}^n   (Y_i - \sum_{k\neq j} \langle \hat{\beta}_{k,i}^{[t-1]}, K(X_{ik}^{(m)}) \rangle_{\mathbb{K}} ) K(X_{ij}^{(m)})$, where $\hat{\beta}_{k,i}^{[t-1]}$
are the estimates obtained in the previous iteration, and fix the other coordinates.  Then, we can separate the contribution of the $j^{th}$ predictor and write 
\begin{equation*}
\begin{aligned}
\frac{\partial}{\partial \beta_j} L_\lambda 
\left(\boldsymbol{\beta} \right)
&=  - \check{\beta}_j + \hat{\Gamma}_{jj}^{(m)} \left( \beta_j\right) \\
& + \lambda \tilde{w}_j
           \begin{cases}
          \left\|\beta_j\right\|_{\mathbb{K}}^{-1} \beta_j & \text{if } \beta_j \neq 0, \\
       \left\lbrace h\in \mathbb{K}^{(m)} : \left\|h \right\|_\mathbb{K}\leq 1 \right\rbrace & \text{if } \beta_j = 0.
            \end{cases}
\end{aligned}
\end{equation*}
Only for implementation, we assume that $\hat{\Gamma}_{jj}^{(m)}$ is equal to the identity matrix. This is equivalent to a well-known \emph{group-wise orthonormality condition} in the group lasso \citep{roche2023}. %This assumption is stronger than Assumption~\ref{Assumption 1} because it directly implies that the operator $\hat{\Gamma}_{jj}^{(m)}$ is bounded and {\color{red}$ \|\hat{\Gamma}_{jj}^{(m)}\|_{op} \leq \|K\|_{op}.$} 
We then set the last equation to zero to find the minimum of $L_{\lambda}(\boldsymbol{\beta})$, obtaining %Therefore, $\hat{\beta}_j$ will be a minimizer of $L_{\lambda}(\boldsymbol{\beta})$ if it satisfies the equation
\begin{equation}\label{min}
 \hat{\beta}_j = 
\begin{cases} 
0 & \text{if } \left\| \check{\beta}_j \right\|_\mathbb{K} \leq \lambda \tilde{w}_j, \\ 
  \left( 1 + \left( \lambda \tilde{w}_j\right) \left\| \hat{\beta}_j \right\|_{\mathbb{K}}^{-1} \right)^{-1} \check{\beta}_j & \text{if } \left\| \check{\beta}_j\right\|_\mathbb{K} > \lambda \tilde{w}_j. 
\end{cases}
\end{equation}
The first case in \eqref{min} provides an upper bound for $\lambda$. If $\lambda \geq (\tilde{w}_j n)^{-1} \|\sum_{i=1}^{n} Y_i K(X_{ij}^{(m)})\|_{\mathbb{K}}$ for all  $j=1 ,\cdots, p,$ then all predictors are discarded. Therefore, to select the grid of $\lambda$, we start with this value and decrease it in logarithmically spaced intervals at a prespecified rate until we reach the desired number of values for $\lambda$. The second case in \eqref{min} enables us to update non-zero coefficients. For $\hat{\beta}_j \neq 0$, we have
%\[
$
 \hat{\beta}_j  = ( 1+ \frac{\left(\lambda \tilde{w}_j\right)} {\| \hat{\beta}_j \|_{\mathbb{K}}} )^{-1} \check{\beta}_j.
$
%\]
This equation does not have a closed form, even for the scalar design. The problem is that the norm of $\hat{\beta}_j$ on the right-hand side is unknown. To overcome this issue, we compute the $\mathbb{K}$-norm and apply Parseval’s identity
%\[
$
\|\hat{\beta}_j\|_{\mathbb{K}}^2  = \sum_{l=1}^{m} \langle ( 1 + \frac{\lambda \tilde{w}_j}{\|\hat{\beta}_j\|_{\mathbb{K}}})^{-1} \check{\beta}_j, v_l\rangle^2_{\mathbb{K}} = ( 1 + \frac{\lambda \tilde{w}_j}{\|\hat{\beta}_j\|_{\mathbb{K}}})^{-2} \sum_{l=1}^{m} \langle \check{\beta}_j, v_l\rangle^2_{\mathbb{K}}
$,
%\]
obtaining
\[
1 = \frac{\sum_{l=1}^{m} \left\langle \check{\beta}_j, v_l\right\rangle^2_{\mathbb{K}}}{\left( \|\hat{\beta}_j\|_{\mathbb{K}} + \lambda \tilde{w}_j\right)^2}.
\]
We use the NLOPT\_LN\_COBYLA algorithm from the \texttt{R} package \texttt{nloptr} to solve this nonlinear equation for $\|\hat{\beta}_j\|_{\mathbb{K}}$. Then, we use a coordinate descent algorithm to iteratively compute estimates $\hat{\boldsymbol{\beta}}^{[t]}$ for $t = 1,2, \ldots$ that converge to $\hat{\boldsymbol{\beta}}$. 
The iterative algorithm stops when the improvement in the $\mathbb{K}$-norm of the estimate of $\boldsymbol{\beta}^*$ reaches a pre-specified threshold. 
In particular, as stopping criteria, we use a maximum number of iterations of the nonlinear algorithm, a threshold on the improvement of the estimation error, and a kill switch, which means that the process stops when the number of selected variables exceeds a given number. Finally, we employ 5-fold cross-validation to determine the optimal value of $\lambda$. 

As explained above, we run the algorithm twice. First, we set all the weights $\tilde{w}_j$ to 1 to obtain an initial estimate of the coefficients. Then, for those that have a non-zero norm, we set $\tilde{w}_j = 1/ \| \tilde{\beta}_j^{(m)}\|_{\mathbb{K}}$ and run the algorithm a second time to improve the variable selection and refine the estimation.

The operator $K$ can be specified either directly through an orthonormal basis $\{v_l\}_{l=1}^{\infty}$ of $\mathbb{H}$ and a positive non-increasing summable sequence $\{\theta_l\}_{l=1}^{\infty}$ (see Section~\ref{sec:back}) or, when $\mathbb{H}=L^2$, as an integral operator associated with a kernel function $k(t,s)$, i.e., $(Kf)(t) = \int k(t,s) f(s)\,ds$, in which case $\mathbb{K}$ is the RKHS induced by $k(t,s)$.

In the literature, two common examples are the Gaussian kernel 
$ k_{\infty}\left(t, s\right) = \exp\left( - \left|t - s\right|^2/\rho \right)$ and the exponential or Laplacian kernel $k_{1/2}\left(t, s\right) = \exp\left( - \left|t - s\right|/\rho \right)$,
where the parameter $\rho$ controls the length scale. A lower value of $\rho$ results in a kernel that decays more rapidly as the distance \( |t - s| \) increases, leading to functions that exhibit more rapid variations over short distances. The Gaussian kernel $k_{\infty}\left(t, s\right)$ is infinitely differentiable and is well suited for modeling very smooth data. The exponential or Laplacian kernel $k_{1/2}\left(t, s\right)$ possesses only one continuous derivative, making it more appropriate for modeling rough data. These kernels belong to a broader class known as Matérn kernels, defined by
\[
k_{\nu}\left(t, s\right) = \frac{2^{1-\nu}}{\Gamma(\nu)} \left( \frac{\sqrt{2\nu} \left|t - s\right|}{\rho} \right)^\nu K_\nu \left( \frac{\sqrt{2\nu} \left|t - s\right|}{\rho} \right),
\]
where \(\Gamma\) is the Gamma function and \( K_\nu \) denotes the modified Bessel function of the second kind \citep{genton2001}. 
When prior knowledge about the regularity of the coefficient functions is available, the operator \(K\) can be chosen accordingly; otherwise, an appropriate operator and its parameter can be selected by cross-validation.

The predictors are expanded on the eigenfunctions of the operator $K$. To determine the number of basis functions, we use the minimum of $n$ and the smallest number of eigenvalues that satisfy
$ \sum_{l=1}^{m} \theta_l \geq \alpha \sum_{l=1}^{\infty} \theta_l$, which mimics the $100\alpha$\%  explanation of data variability in FPCA. Since we generally do not want to perform dimension reduction at this stage, we recommend using a large $\alpha$, e.g.,~$\alpha=0.99$.

The R code implementing SOFIA, as well as an example with simulated data, is provided 
%in a GitHub page.
on the GitHub page \verb|https://github.com/HedayatFathi/SOFIA|.

\section{Simulation study}\label{sec:sim}
%In this section, we conduct a comprehensive simulation study to evaluate SOFIA's performance. We consider both the accuracy of variable selection and the predictive performance in several scenarios. We generate datasets under different model complexities and signal-to-noise ratios, and compare SOFIA with existing methods under various specifications. 

%\subsection{Scenario 1}
Although SOFIA applies to any separable Hilbert space,
in order to have a fair comparison with the competing methods, we restrict the main simulation to $L^2[0,1]$. To illustrate our method and implementation in their generality, in Section~S2.6 
%of the Supplementary material,
we consider the same simulation settings in the Sobolev space $\mathbb{H} = \mathbb{H}^1[0,1]$, with an operator $K$ that is directly generated by a spectral decomposition.

We take the simulation framework from \cite{gertheiss2013} and introduce some modifications in order to obtain scaled design matrices and the curve domain $[0,1]$. In particular, we set $\mathbb{H} = L^2[0,1]$ and define the functional predictors as
$$X_{ij} \left( t \right) = 0.01\! \left( \sum_{r=1}^{5} \left( a_{ijr} \sin \left( 2 \pi t (5\! -\! a_{ijr}) \right) \!-\! m_{ijr} \right) + 15 \right),$$ 
where $a_{ijr}$ and $m_{ijr}$ are independently sampled from the uniform distributions $U(0, 5)$ and $U(0, 2\pi)$, respectively.
We then generate the response according to the model 
\[Y_i = \sum_{j=1}^{p_0} \int_{0}^{1} \beta_j(t) X_{ij}(t) \, dt + \varepsilon_i,\]
where $\varepsilon_i \sim \mathcal{N}(0, \sigma^2)$.
 We examine the robustness of these distributional choices in Sections~S2.7 and S2.8. % of Supplementary material
In particular, in Section~S2.7 we replace the Gaussian error with a Gaussian mixture (i.e.,~a sub-Gaussian noise), while in Section~S2.8 we replace the uniform random variables used in generating the predictor curves with Beta \((2,5)\)-distributed variables. Tables~S5 and S6 show that SOFIA's performance is robust to the particular choice of distribution for errors and predictors, confirming the observations reported in this section.
We fix the sample size at $n=500$ and discretize the interval $[0,1]$ using a grid of 50 equidistant points. 

The number of predictors $p$ and of active predictors $p_0$ %, indexed by $j$, 
vary between the different experimental scenarios. 
In particular, we consider \textbf{high-dimensional} and \textbf{moderate} scales, and \textbf{very sparse} and \textbf{less sparse} regimes for each scale. 
In the high-dimensional scale, we include $p=700$ functional predictors, which is particularly challenging due to the large number of features relative to the sample size $(n < p)$. Notably, among the competing methods, only FAStEN \citep{boschi2024} and fEnet \citep{GuoLiHsing2028} can handle this setting, whereas the others are computationally infeasible. 

In the high-dimensional scale, the very sparse regime considers five active predictors ($p_0=5$), while the less sparse regime includes 10 active predictors ($p_0=10$). 
In the moderate scale, the very sparse regime contains $p=30$ predictors of which $p_0=10$ are active, while the less sparse regime has $p=10$ predictors of which $p_0=5$ are active. This moderate scale setting enables broader comparisons with six competing methods. 
The coefficient functions $\beta_j$ of the active predictors are simulated based on the Gamma and exponential density functions (see Figure~S1).% in the Supplementary material). 

%The first one is already considered in \cite{gertheiss2013}, with 10 predictors such that the first five are active. In the second scenario, we increase the number of predictors to 30, and we set 10 of them to be active. 

For each scenario, we consider five levels of signal-to-noise ratio, namely $\text{SNR} = 0.1, 0.5, 1, 10$, and $100$.  The variance of the noise term $\varepsilon_i$ is set according to the desired SNR level by $\sigma^2= \operatorname{var}(\mathbf{Y}^{true})/\mathrm{SNR}$. 
Here, the noise-free scalar response is defined as $Y^{true}_i = Y_i - \varepsilon_i$, and $var(\boldsymbol{Y}^{true})$ denotes its variance over $i = 1, \ldots, n$. 
%The signal-to-noise ratio (SNR), as usual, is defined 
%\[
%\text{SNR} = \frac{\text{Signal Variance}}{\text{Noise Variance}}.
%\]
%The signal variance for centered predictors \(X_{ij}(t)\) is computed by
%\[
%\text{Signal Variance} = \text{Var} \left( \sum_{j=1}^{p} \int_{0}^{1} \beta_j(t) X_{ij}(t) \, dt \right).
%\] 
%We have considered standardized \(X_{ij}(t)\) with variance 1. So, it becomes
%\[
%\text{Signal Variance} = \sum_{j=1}^{p} \int_{0}^{1} \beta_j^2(t) dt.
%\]
%The noise Variance is
%\[
%\text{Noise Variance} = \sigma^2.
%\]

The performance of the methods is evaluated based on the average number of true positives (TP) and false positives (FP) across 50 replications, as well as %  In addition, we assess prediction accuracy by computing 
the root mean squared error (RMSE) on an independently generated test set
%, namely
%$
%\sqrt {\sum_{i=1}^{n_{\text{test}}}  (Y_i - \hat{Y}_i)^2}/n_\text{test}
%$
of size $n_{\text{test}} = 100$. Since TP and FP measure complementary aspects of variable selection, Section~S2.3 %of the Supplementary material
reports the F1 score and the exact support recovery rate in the high-dimensional setting as joint measures of selection performance (Table~S1).

To choose the tuning parameter $\lambda$ in SOFIA, we examine 100 candidate values spaced logarithmically with a ratio of 0.05 (see Section~\ref{sec:impl}), using a 5-fold cross-validation.To determine the number of basis functions $m$, we employ $\alpha=0.99$. The sensitivity of SOFIA to this choice is examined in Section~S2.4 (Table~S2), % of Supplementary material
using cumulative spectral-mass thresholds of 80\%, 90\%, 95\%, 99\%, and 99.5\%. Results show that the variable selection performance of SOFIA is robust to the choice of $m$.

\subsection{High-dimensional scale}\label{subsec:high}

We first compare the performance of SOFIA when using different kernels and different kernel parameters (see Section~S2.1). % of the Supplementary material). 
For high signal-to-noise ratios (SNR = 100 and 10), SOFIA shows perfect support recovery for all kernels. 
In the noisier case (SNR = 1), SOFIA continues to show very good performance, but sometimes misses an active predictor and/or includes a false positive. 
In the high-noise case of SNR = 0.5, performance is still pretty good, while it deteriorates markedly for the extremely noisy case of SNR = 0.1.
Interestingly, all kernel functions exhibit comparable performance, although the Gaussian kernel has slightly higher prediction errors. Examples of estimated coefficient functions $\hat{\boldsymbol{\beta}}$ for different kernels are displayed in Section~S2.2.% of the Supplementary material. 

Second, we compare the performance of SOFIA (with Gaussian kernel and $\rho =8$, FAStEN \citep{boschi2024}, and fEnet \citep{GuoLiHsing2028} across different SNR levels and both sparsity levels. For FAStEN and SOFIA, we use a kill switch of $2p_0$ predictors. For fEnet, we use the parameter values and tuning grid
provided in the original implementation of \cite{GuoLiHsing2028}, without imposing an additional kill switch. The sensitivity to the kill switch is examined in Section~S2.5, %of Supplementary material
where for FAStEN and SOFIA the maximum model size is set to \(2p_0\), \(10p_0\), or left unrestricted (Table~S3). Results show that this kill switch does not change the performance of SOFIA, except slightly in the very noisy case of SNR = 0.1. We note that FAStEN fails to converge when the maximum model size is left unrestricted.
\begingroup
\renewcommand{\arraystretch}{0.7}
\setlength{\tabcolsep}{4pt}

\begin{table}[ht]
    \small
    \centering

    \rowcolors{3}{gray!10}{white}
    \begin{tabular}{c c c c c c c c}

        \toprule
        \multicolumn{2}{c}{} 
        & \multicolumn{3}{c}{\textbf{Very sparse ($p_0=5$)}} 
        & \multicolumn{3}{c}{\textbf{Less sparse ($p_0=10$)}} \\
        
        \cmidrule(lr){3-5} 
        \cmidrule(lr){6-8}
        
        \textbf{SNR} & \textbf{Method} 
        & \textbf{TP} & \textbf{FP} & \textbf{RMSE} 
        & \textbf{TP} & \textbf{FP} & \textbf{RMSE} \\
        \midrule

        & FAStEN & 1.2 & 0.6 & 3.88 
                  & 1.0 & 1.0 & 5.75 \\
        0.1 
        & fEnet  & 0.1 & \textbf{0.3} & 4.01 
                  & 0.0 & \textbf{0.1} & 5.18 \\
        & SOFIA  & \textbf{1.8} & 2.2 & \textbf{3.87} 
                  & \textbf{2.0} & 4.0 & \textbf{5.16} \\
        \midrule

        & FAStEN & \textbf{4.9} & 0.3 & 1.91 
                  & \textbf{8.0} & 2.9 & 2.90 \\
        0.5
        & fEnet  & 4.0 & \textbf{0.1} & 1.89 
                  & 2.1 & 1.8 & 2.81 \\
        & SOFIA  & 4.7 & \textbf{0.1} & \textbf{1.83} 
                  & 7.7 & \textbf{1.3} & \textbf{2.47} \\
        \midrule

        & FAStEN & \textbf{5.0} & \textbf{0.0} & 1.65 
                  & \textbf{9.7} & 2.3 & 2.32 \\
        1
        & fEnet  & 3.5 & \textbf{0.0} & 1.42 
                  & 5.8 & \textbf{0.1} & 1.98 \\
        & SOFIA  & \textbf{5.0} & \textbf{0.0} & \textbf{1.38} 
                  & 9.0 & 0.3 & \textbf{1.88} \\
        \midrule

        & FAStEN & \textbf{5.0} & \textbf{0.0} & 1.23 
                  & \textbf{10.0} & \textbf{0.0} & 1.62 \\
        10
        & fEnet  & \textbf{5.0} & \textbf{0.0} & \textbf{0.55} 
                  & \textbf{10.0} & \textbf{0.0} & \textbf{0.79} \\
        & SOFIA  & \textbf{5.0} & \textbf{0.0} & 0.78 
                  & 9.9 & \textbf{0.0} & 1.08 \\
        \midrule

        & FAStEN & \textbf{5.0} & \textbf{0.0} & 1.18 
                  & \textbf{10.0} & \textbf{0.0} & 1.53 \\
        100
        & fEnet  & \textbf{5.0} & \textbf{0.0} & \textbf{0.40} 
                  & \textbf{10.0} & \textbf{0.0} & \textbf{0.58} \\
        & SOFIA  & \textbf{5.0} & \textbf{0.0} & 0.70 
                  & \textbf{10.0} & \textbf{0.0} & 0.96 \\

        \bottomrule
    \end{tabular}

    \caption{Comparison between FAStEN, fEnet, and SOFIA across different SNR levels in the high-dimensional setting ($p=700$). Boldface indicates the best-performing method for each metric.}
    \label{tab:comp_high-dimensional}
\end{table}
\endgroup

Table~\ref{tab:comp_high-dimensional} reveals a clear dependence on the signal-to-noise ratio. In the low-noise settings (SNR = 10 and 100), all three methods recover the true support almost perfectly, with zero false positives in both sparsity regimes. In these settings, fEnet achieves the lowest RMSE, whereas SOFIA also improves on FAStEN in terms of prediction error while maintaining nearly perfect support recovery. As the noise increases (SNR=0.5 and 1), SOFIA generally recovers more active predictors than fEnet, particularly in the
less-sparse regime, and achieves the lowest RMSE among the three methods.
Although fEnet often produces few false positives, this is accompanied by lower true-positive recovery in several settings. At SNR = 0.1, all methods face substantial difficulty. fEnet selects very few or no active predictors, whereas SOFIA recovers more active predictors, generally at the cost of additional false positives.
Overall, fEnet is particularly competitive in low-noise settings, while SOFIA and FAStEN provide a more favorable support-recovery trade-off in noisy settings. Joint measurements of the variable selection performance, containing F1 score and exact support recovery, are provided in Section~S2.3. %of the Supplementary material

\subsection{Moderate scale}

\begingroup
\renewcommand{\arraystretch}{0.7} % Adjust to 1.0 for single spacing
\setlength{\tabcolsep}{4pt} % default is usually 6pt
\begin{table}[] 
\centering
\small

\rowcolors{3}{gray!10}{white}
\begin{tabular}{c c c c c c c c}
    \toprule
    \multicolumn{2}{c}{} & \multicolumn{3}{c}{\textbf{Very sparse}} & \multicolumn{3}{c}{\textbf{Less sparse}} \\
    \multicolumn{2}{c}{} & \multicolumn{3}{c}{\textbf{($p=30, ~p_0=10$)}} & \multicolumn{3}{c}{\textbf{($p= 10, ~p_0=5$)}} \\
    \cmidrule(lr){3-5} \cmidrule(lr){6-8}
    \textbf{SNR} & \textbf{Method} & \textbf{TP} & \textbf{FP} & \textbf{RMSE} & \textbf{TP} & \textbf{FP} & \textbf{RMSE} \\
    \midrule
    & Standard & 8.1 & 7.9 & \textbf{4.95} & \textbf{4.8} & 2.8 & 3.74 \\
    & Adapt1   & \textbf{8.2} & 7.7 & 5.09 & \textbf{4.8} & 3.2 & 3.9 \\
    & Adapt2   & 8.1 & 6.9 & 5.05 & \textbf{4.8} & 2.2 & \textbf{3.68} \\
    0.1 & Roche   & 0 & \textbf{0} & 5.10 & 0 & \textbf{0} & 3.86 \\
    & fEnet   & 0.8 & 0.6 & 5.22 & 0.8 & 0.2 & 4.04 \\
    & FAStEN   & 2.2 & 0.3 & 5.76 & 3.1 & 0.3 & 4.41 \\
    & SOFIA    & 7.1 & 3.7 & 4.99 & 4.5 & 0.8 & 3.84 \\
    \midrule
    & Standard & \textbf{10} & 11.3 & 2.48 & \textbf{5} & 2.6 & \textbf{1.69} \\
    & Adapt1   & \textbf{10} & 10.4 & 2.58 & \textbf{5} & 2.6 & 1.78 \\
    & Adapt2   & \textbf{10} & 7.8 & 2.54 & \textbf{5} & 1.8 & 1.71 \\
    0.5 & Roche   & 0 & \textbf{0} & 2.67 & 0 & \textbf{0} & 2.04 \\
    & fEnet   & 5.2 & 0.4 & 2.65 & 4.3 & 0.3 & 1.89 \\
    & FAStEN   & 9.4 & 0.9 & 2.92 & \textbf{5} & 0.2 & 2.24 \\
    & SOFIA    & 9.5 & 0.7 & \textbf{2.37} & \textbf{5} & \textbf{0} & 1.83 \\
    \midrule
    & Standard & \textbf{10} & 10.8 & 1.96 & \textbf{5} & 2.2 & \textbf{1.24} \\
    & Adapt1   & \textbf{10} & 10.9 & 2.06 & \textbf{5} & 2.1 & 1.34 \\
    & Adapt2   & \textbf{10} & 6.74 & 2.01 & \textbf{5} & 1.7 & 1.26 \\
    1 & Roche    & 0 & \textbf{0} & 2.18 & 0 & \textbf{0} & 1.67 \\ 
    & fEnet   & 6.4 & 0.5 & 1.98 & 4.7 & 0.9 & 1.39 \\
    & FAStEN   & \textbf{10} & 0.6 & 2.33 & \textbf{5} & 0.1 & 1.79 \\
    & SOFIA    & 9.8 & 0.1 & \textbf{1.80} & \textbf{5} & \textbf{0} & 1.38 \\
    \midrule %\addlinespace
    & Standard & \textbf{10} & 4.7 & 1.31 & \textbf{5} & 0.6 & 0.59 \\
    & Adapt1   & \textbf{10} & 7.1 & 1.44 & \textbf{5} & 0.8 & 0.79 \\
    & Adapt2   & \textbf{10} & 3.76 & 1.41 & \textbf{5} & 0.4 & 0.63 \\
    10 & Roche    & 9.54 & \textbf{0}  & 1.22  & \textbf{5} & \textbf{0} & 1.03 \\
    & fEnet   & \textbf{10} & \textbf{0} & \textbf{0.77} & \textbf{5} & \textbf{0} & \textbf{0.55} \\
    & FAStEN   & \textbf{10} & 0.2 & 1.64 & \textbf{5} & \textbf{0} & 1.25 \\
    & SOFIA    & \textbf{10} & \textbf{0} & 1.04 & \textbf{5} & \textbf{0} & 0.78 \\
    \midrule %\addlinespace
    & Standard & \textbf{10} & 3 & 1.22 & \textbf{5} & 0.3 & 0.48 \\
    & Adapt1   & \textbf{10} & 4.3 & 1.37 & \textbf{5} & 0.4 & 0.75 \\
    & Adapt2   & \textbf{10} & 1.8 &  1.35& \textbf{5} & 0.3 & 0.53 \\ 
    100 & Roche    & \textbf{10}  & \textbf{0} & \textbf{0.56} & \textbf{5} & \textbf{0} & 1.00 \\ 
    & fEnet   & \textbf{10} & \textbf{0} & \textbf{0.56} & \textbf{5} & \textbf{0} & \textbf{0.39} \\
    & FAStEN   & \textbf{10} & \textbf{0} & 1.56 & \textbf{5} & \textbf{0} & 1.18 \\
    & SOFIA    & \textbf{10} & \textbf{0} & 0.86 & \textbf{5} & \textbf{0} & 0.69 \\
    \bottomrule
\end{tabular}
\caption{Comparison of SOFIA with six competitor methods across different SNR levels in the moderate scale. Boldface indicates the best-performing method for each metric.}
\label{tab:small_comparition} 
\end{table}
\endgroup

Since most existing methods are computationally infeasible for a large number of predictors and/or cannot handle cases where the number of predictors exceeds the sample size (such as the high-dimensional case shown above), we consider a moderate scale simulation setting to allow for a broader performance comparison. 
We compare the performance of SOFIA (with Gaussian kernel and $\rho=8$) with six other methods: the standard non-adaptive group lasso \citep{fan2015}, adaptive lasso \citep[Adapt1 and Adapt2;][]{gertheiss2013}, the method of \cite{roche2023}, fEnet of \cite{GuoLiHsing2028}, and FAStEN \citep{boschi2024}. To ensure a fair comparison, all methods are run without a kill switch and with their own default configuration parameters.
Table~\ref{tab:small_comparition} summarizes the results. 
In low-noise conditions (SNR = 100 and 10), all methods perform generally well. In particular, SOFIA, FAStEN, fEnet, and Roche consistently achieve high true positive recovery with minimal false positives across both regimes. Although standard and adaptive lasso sometimes achieve better RMSE and always select all true active predictors, their elevated false-positive rates make them less suitable for tasks that require precise variable selection and parsimonious models. 
In high-noise conditions (SNR = 0.5 and 1), SOFIA and FAStEN effectively balance selection and precision, achieving very high true positive rates with relatively low false positives. In contrast, the variable selection performance of fEnet is noticeably worse, while Roche does not select any feature (hence showing both zero true and zero false positive rates). Both standard and adaptive lasso correctly select all true active predictors but also include several false positives, with Adapt2 performing slightly better than Adapt1 and standard lasso in both regimes.
Finally, the extremely noisy case of SNR = 0.1 is challenging for all methods. In particular, fEnet selects very few predictors and recovers only a small fraction of the active set, whereas SOFIA recovers a substantial proportion of the active predictors and also achieves a lower RMSE. FAStEN shows performance in between fEnet and SOFIA for true and false positives, but a higher RMSE. 

\section{Real data analysis}\label{sec:real}

In this section, we apply SOFIA to a real-world economic dataset to identify the most relevant functional predictors to predict future GDP growth in Canada. 
%Our goal is twofold: first, to assess SOFIA's capacity to select meaningful variables in a high-dimensional real dataset; second, to evaluate the predictive performance of SOFIA compared to classic forecasting methods and classic functional regression. 
The time-varying nature of economic and financial data makes FDA highly relevant for their analysis. One challenge in these research areas is the high frequency of the data, which can lead to overfitting, noise, and computational complexity. FDA offers robust methods to handle and analyze high-frequency data by smoothing and leveraging their continuous nature \citep{zhang2023}. Another challenge is the heterogeneity of the data time frequencies: some variables are reported monthly, while others are reported quarterly or annually \citep{kapetanios2016}. In FDA, different measurement frequencies are easily managed because all variables are transformed into curves on a continuous domain and can then be evaluated on the same grid. For recent applications of FDA in economics and finance see, e.g.,~\cite{reimherr2018,severino2022,cremona2023}.

The prediction of GDP growth is crucial for macroeconomic analysis, as it helps policymakers, financial institutions, and businesses to make decisions. We are interested in cases where multiple macroeconomic predictors are observed, and we want to determine the relevant variables to predict GDP growth. %Such studies exist in the literature, for example, \cite{stock2002}, which analyzed 215 monthly time series variables representing a wide range of macroeconomic indicators across the United States. We aim to utilize the facilities the FDA provides to select relevant variables and predict GDP growth using these relevant variables.  
We use a dataset from \cite{fortin2022}, which provides a comprehensive macroeconomic database on Canadian economic indicators. We consider a balanced version of the data that overcomes the issues related to stationarity, missing data, and discontinuities. We use monthly data from January 1981 to December 2019, resulting in 468 monthly observations. This time frame excludes the unpredictable outbreak of the COVID-19 pandemic.
The dataset includes many international, national, and provincial economic variables. To ensure meaningful and interpretable results, we restrict our analysis to national-level indicators and select 20 variables that are commonly used in GDP prediction studies. These variables represent various financial sectors and include %predictors such as 
total GDP growth (\texttt{GDP\_new}), industrial production GDP growth (\texttt{IP\_new}), retail trade GDP (\texttt{RT\_new}) growth, public administration GDP growth (\texttt{PA\_new}), oil production  (\texttt{OIL\_CAN\_new}), employment rate (\texttt{EMP\_CAN}), unemployment rate (\texttt{UNEMP\_CAN}), new housing price index (\texttt{NHOUSE\_P\_CAN}), housing starts units (\texttt{Hstart\_CAN\_new}), total building permits (\texttt{build\_Total\_CAN\_new}), M3 money supply (\texttt{M3}), household credit (\texttt{CRED\_HOUS}), business credit (\texttt{CRED\_BUS}), bank rate (\texttt{BANK\_RATE\_L}),  three-month Treasury bill rate (\texttt{TBILL\_3M}), total Canada’s official international reserves (\texttt{RES\_TOT}), total imports (\texttt{Imp\_BP\_new}), total exports (\texttt{Exp\_BP\_new}), exchange rate between the US and the Canadian dollar (\texttt{USDCAD\_new}), and the inflation rate measured by the Consumer Price Index (\texttt{CPI\_ALL\_CAN}). 
A complete list and description of the variables is provided in Section~S3 (Table~S7). % of Supplementary material.
 For more information on the variables, see \cite{fortin2022}.

We construct functional predictors using monthly data by considering 12 consecutive values as a curve. We do not apply additional smoothing, so that the functional representation preserves the original structure of the data. Indeed, since the data have been transformed in the original paper to reach stationarity, smoothing would artificially flatten the curves, filtering out important information. 
To generate enough observations, we employ a rolling window method with a one-month shift. Precisely, each functional observation corresponds to a sequence of 12 consecutive months: the first observation consists of a curve obtained by the data from months 1 to 12, the second from months 2 to 13, and so on. Following this procedure, we obtain a dataset that contains 456 observations and 20 predictors. These successive functional observations are strongly dependent, as consecutive curves share 11 of their 12 monthly values. Consequently, this setting is not covered by our theoretical framework. Extending the oracle property to serially dependent errors remains for future work. 
For the scalar response, we consider the GDP growth of the month immediately following the 12-month predictors. GDP growth is measured as the change in GDP from one month to the next, using the difference of logarithmic values \cite[\texttt{GDP\_new} in][]{fortin2022}. 

Since we want to compare the predictive performance of SOFIA with other methods, we designate the final 36 months of the dataset as the test set for prediction comparison, while the remaining 420 observations are used for fitting the models.
We run SOFIA with an exponential kernel and parameter 2, which is well-suited for handling rough data (the functional predictors are not smoothed) and exhibits good robustness to noise. For the selection of $\lambda$, we use a grid of 100 candidate values with a ratio of 0.1, and a rolling-window cross-validation. Specifically, the first 315 observations (75\% of the data) are used as the initial training set, while the next 21 observations (5\% of the data) are used as the test set. In each subsequent fold, we shift both the training and test sets forward by 21 observations (5\% of the data). This procedure results in a 5-fold rolling-window cross-validation framework. 

\begin{figure*}[bp] 
    \centering
    \includegraphics[width=0.9\textwidth]{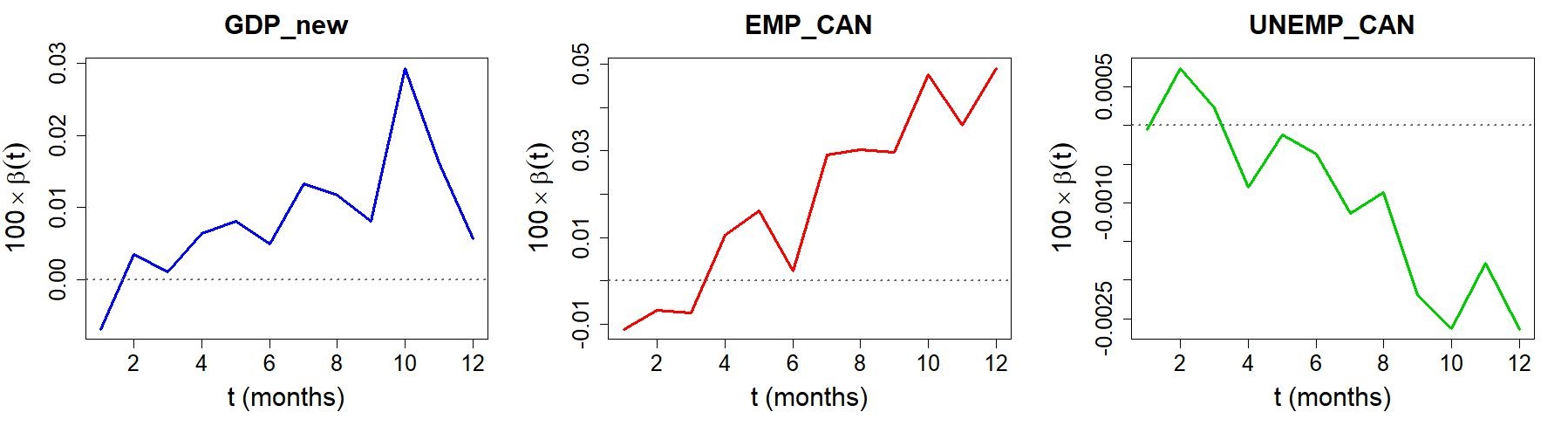} 
    \caption{Estimated coefficients ($\hat{\boldsymbol{\beta}}$) for the three selected predictors for GDP growth.}
    \label{fig:GDP} 
\end{figure*}

SOFIA selects three functional predictors, namely the total GDP growth (\texttt{GDP\_new}), employment rate (\texttt{EMP\_CAN}), and unemployment rate (\texttt{UNEMP\_CAN}).  
%Although no literature explicitly identifies employment and unemployment rates as the most important predictors of GDP growth, several studies emphasize a significant relation between GDP and these variables, often in the context of Okun’s Law. This economic principle, first introduced by Arthur Okun in the 1960s, describes an inverse relation between GDP growth and unemployment rates. Specifically, a 1\% decrease in GDP is typically associated with an approximate 3 percentage-point increase in unemployment (see, for example, \cite{hashmi2021} and references therein). 
Figure~\ref{fig:GDP} presents the estimated $\hat{\boldsymbol{\beta}}$ coefficients corresponding to the selected predictors. The estimated coefficient functions reveal that GDP growth and the employment rate have a positive relation with the GDP growth of the next month, and their effects grow stronger over time until the 10\textsuperscript{th} month. This suggests that, in the prediction of GDP growth, the GDP growth rate observed three months before has the highest predictive power relative to the other months. In contrast, the unemployment rate shows a negative effect, indicating that higher unemployment is associated with lower GDP growth. Similar to the other two selected variables, the values of the three months preceding the forecast month are the most influential for this predictor. Interestingly, at the beginning of the period, all three predictors show a pattern opposite to what is observed later. This may suggest a form of mean reversion, where early movements are gradually offset or reversed as the economy adjusts over time.

Since the data have a time-series structure, the competing functional variable selection methods in Section~\ref{sec:sim} would require additional adaptation to preserve the temporal evaluation scheme. Instead, we compare the performance of SOFIA with three classic time series models as benchmarks: the AutoRegressive Integrated Moving Average (ARIMA), the AutoRegressive Integrated Moving Average with Exogenous Regressors \citep[ARIMAX;][]{gu2022}, and the naïve model that forecasts the next GDP growth as the last observed value of GDP growth. We also consider a simple scalar-on-function regression model (without penalization) that uses only \texttt{GDP\_new} as a functional predictor.
We estimate ARIMA and ARIMAX using the function \texttt{auto.arima()} from the \texttt{R} package \textbf{forecast} (with default parameters), which automatically selects the optimal parameters based on the corrected Akaike Information Criterion. 
%We fit the models with the default settings, including the maximum values for the autoregressive and moving average orders. 
Then, we give the fitted models as input to the \texttt{Arima()} function to forecast GDP growth for the following month on the test set without reestimating the model. In ARIMAX, we include as exogenous variables all the predictors employed by SOFIA at time \( t - 1 \), ensuring a fair comparison in terms of the variables considered. %Finally, we implement the Simple Functional Regression model as a minimalist functional approach, incorporating only the functional representation of \texttt{GDP\_new} over the 12-month predictor window, without any penalization or variable selection. 
The simple scalar-on-function model is implemented using SOFIA without penalization (i.e.,~$\lambda=0$). 

\begingroup
\renewcommand{\arraystretch}{0.7} % Adjust to 1.0 for single spacing
\setlength{\tabcolsep}{3pt} % default is usually 6pt
\begin{table}[tbp]
    \centering
     
    \begin{tabular}{lccccc}
        \toprule
        \textbf{Model} &  \textbf{ARIMA} & \textbf{ARIMAX}& \textbf{Naïve} & \textbf{Simple SoF} & \textbf{SOFIA} \\
        \midrule
           \textbf{RMSE}  & 0.2163 & 0.2361& 0.2625 & 0.2234 & \textbf{0.2125}\\
         \textbf{MAE}  & 0.1781  & 0.1742 & 0.1939 &0.1830	 & \textbf{0.1651} \\
        \bottomrule
    \end{tabular}
    \caption{Test RMSE and MAE ($\times$ 100) for the different methods.}
    \label{tab:rmse_comparison}
 \end{table}
\endgroup

Table~\ref{tab:rmse_comparison} shows the test RMSE and mean absolute value (MAE) for the five compared methods over the 36 months in the test set. %For readability, all RMSE and MAE values in the table are multiplied by 100. 
The selected ARIMA model is an ARIMA(4,0,5) with a nonzero mean, %It includes p= 4 autoregressive (AR) terms, No differencing was needed to make the series stationary (d=0), It includes q = 5 moving average (MA) terms, and the model has a nonzero intercept.
while the selected ARIMAX model has ARIMA(3,0,0) errors. 
We observe that SOFIA achieves the lowest prediction error among the considered methods, slightly outperforming ARIMA in both RMSE and MAE. Although the ARIMAX model shows a marginal improvement in MAE compared to ARIMA, its RMSE is higher, suggesting that the inclusion of all exogenous predictors does not consistently enhance forecast accuracy. The simple scalar-on-function regression model performs moderately well but is clearly outperformed by SOFIA, as expected. This highlights the benefit of selecting informative predictors through the SOFIA procedure rather than relying on a single functional covariate.

\section{Conclusions and future work}
\label{sec:concl}
In this paper, we introduced SOFIA for variable selection in scalar-on-function regression models, which enables the estimation of functional coefficients with a desired level of smoothness (or, more in general, regularity). To achieve this goal, we require the coefficients to lie in a specific subspace that determines the smoothness level, while allowing the predictors to reside in a general Hilbert space. We showed that SOFIA satisfies the functional oracle property, even in scenarios where the number of predictors is much larger than the sample size. 
In an extensive simulation study, we show that SOFIA is competitive with recent state-of-the-art methods in variable selection and prediction accuracy. 
%and provides a more favorable support recovery in noisy setting. 
An additional advantage of SOFIA is that it offers direct control over the smoothness, or more generally the regularity, of the estimated coefficient functions through the choice of the operator $K$. 

Our approach adopts a sieve truncation on the eigenbasis of the operator $K$ and solves the penalized least squares with functional subgradients. This strategy yields a sequence of finite-dimensional problems on nested subspaces $\mathbb{K}^{(m)}$ and leads to nonparametric convergence rates that explicitly depend on the sieve level and the complexity of the design.  However, selecting the optimal value for the dimension $m$ remains an open problem. 
%Moreover, the representer theorem offers an alternative that enables the study of analytical properties without truncation. This approach is more closely aligned with nonparametric methods and is left for future work.

The literature in the area of functional variable selection is limited, and there is a great space for future research. For instance, we assume that the functional predictors are fully observed and are not contaminated by noise, whereas, in practice, the variables are noisy and observed on a finite grid. Considering the effect of noise and grid size remains an open issue.  More broadly, SOFIA is a linear regression method with a lasso penalty. As time-varying data acquisition advances, the need for more complex statistical tools (nonlinear methods or more complex penalties) becomes increasingly more important. This evolution introduces numerous theoretical and computational challenges that push the research agenda in FDA and, in particular, in functional regression models. Finally, an important direction for future work is to develop valid post-selection inference for SOFIA, for example, confidence bands or hypothesis tests for the selected coefficient functions that explicitly account for the data-driven selection step.

\textbf{Acknowledgments.} We thank Ana Maria Kenney, Michael Morin, Hans-Georg Müller, Matthew Reimherr, Bharath Sriperumbudur and Ruodu Wang, as well as participants at the Second Workshop Mathematics for Artificial Intelligence and Machine Learning (Università Bocconi, 2024), at the 2024 Annual Meeting of the Statistical Society of Canada (Memorial University of Newfoundland), at the Canadian Mathematical Society Summer Meeting 2025 (Québec), at the 2025 International Conference on Statistics and Data Science (Vancouver), at Université Laval (2025), at the Polytechnic University of Milan (2026) and at the University of Eastern Piedmont (2026) for useful comments.

H. Fathi was partially supported by the Faculty of Business Administration, Université Laval, and by the Interuniversity Research Center on Enterprise Network, Logistics and Transportation (CIRRELT).
M.A. Cremona is the chairholder of the Chair in Statistical Learning and was funded by the Natural Sciences and Engineering Research Council of Canada (NSERC, grant RGPIN-2020-05657), the Fonds de recherche du Québec Santé (FRQS, grant 2023-2024-JC-339901), the Social Sciences and Humanities Research Council (SSHRC, grant 430-2020-00358) and by the Faculty of Business Administration, Université Laval.
F.Severino was funded by the Social Sciences and Humanities Research Council (SSHRC, grant 430-2020-00358).
This research was enabled by the Digital Research Alliance of Canada.

\textbf{Disclosure statement.} The authors report there are no competing interests to declare.

\textbf{Data availability statement.} The data that support the findings of this study are openly available at \\ \verb|https://www.stevanovic.uqam.ca/DS_LCMD.html|.

%%%%%%%%%%%%%%%%%%%%%%%%%%%%%%%%%%%%%%% REFERENCES
	%\clearpage
	%\nocite*
	%\bibliographystyle{plain}
	\bibliographystyle{Chicago}
	\small{\bibliography{references}}

\end{document}